\documentclass{amsart}
\usepackage{amsmath,amsfonts,amssymb,amsthm,mathtools}
\usepackage{amscd}
\usepackage{bbm}
\usepackage{enumerate}
\usepackage{galois}
\usepackage{mathrsfs}
\usepackage{xypic}
\usepackage{geometry}
\usepackage{hyperref}

\usepackage{color}

\theoremstyle{plain}
\newtheorem{Th}{Theorem}[section]

\newtheorem{Lem}[Th]{Lemma}
\newtheorem{Prop}[Th]{Proposition}
\newtheorem{Conj}[Th]{Conjecture}

\newtheorem{Def}{Definition}[section]
\newtheorem{Ex}{Example}[section]
\newtheorem{Rem}{Remark}[section]
\theoremstyle{plain}
\newtheorem*{Main}{Main Theorem}
\numberwithin{equation}{section}

\DeclareMathOperator{\ima}{Im}

\newcommand{\derx}{\mathrm{Der}^\qp(\hm A)}
\newcommand{\der}{\mathrm{Der}(\hm A)}
\newcommand{\diff}[2]{\frac{\partial #1}{\partial #2}}

\newcommand{\qa}{\alpha}
\newcommand{\qb}{\beta}
\newcommand{\qd}{\delta}
\newcommand{\qg}{\gamma}
\newcommand{\qs}{\sigma}
\newcommand{\qt}{\tau}
\newcommand{\qth}{\theta}
\newcommand{\qe}{\varepsilon}
\newcommand{\qz}{\zeta}
\newcommand{\qp}{\partial}

\newcommand{\Qg}{\Gamma}
\newcommand{\ql}{\lambda}

\newcommand{\Qd}{\Delta}
\newcommand{\Qo}{\Omega}

\newcommand{\vard}[2]{\frac{\delta #1}{\delta #2}}

\newcommand{\hm}[1]{\hat{\mathcal #1}}

\newcommand{\kk}[1]{\left(#1\right)}

\newcommand{\fk}[2]{\left[#1, #2\right]}

\begin{document}

\title{Variational Bihamiltonian Cohomologies and Integrable Hierarchies II: Virasoro symmetries}

\author{Si-Qi Liu, Zhe Wang, Youjin Zhang}
\keywords{Virasoro symmetries, Super tau-covers, Variational Bihamiltonian Cohomologies, Integrable Hierarchies}

\begin{abstract}We prove that for any tau-symmetric bihamiltonian deformation of the tau-cover of the Principal Hierarchy associated with a semisimple Frobenius manifold, the deformed tau-cover admits an infinite set of Virasoro symmetries.
\end{abstract}

\date{\today}

\maketitle
\tableofcontents


\section{Introduction}\label{sec1}
This is the second one of the series of papers devoted to the study of the deformations of the Virasoro symmetries of bihamiltonian integrable hierarchies. In the first one \cite{liu2021variational}, we developed a cohomology theory on the space of differential forms of the infinite jet space of a super manifold for a given bihamiltonian structure of the hydrodynamic type, and we call such a cohomology theory the \textit{variational bihamiltonian cohomology}. It can be viewed as a generalization of the bihamiltonian cohomology introduced in \cite{dubrovin2001normal}, and it provides a suitable tool for us to study deformations of Virasoro symmetries of bihamiltonian integrable hierarchies.

The purpose of the present paper is to prove the following theorem.
\begin{Main}
For the Principal Hierarchy associated with a semisimple Frobenius manifold and any of its tau-symmetric bihamiltonian deformations, there exists a unique deformation of its Virasoro symmetries such that they are symmetries of the deformed integrable hierarchy. Moreover, the action of the Virasoro symmetries on the tau-function $Z$ of the deformed integrable hierarchy can be represented in the form 
\begin{equation}
\label{DN}
\diff{Z}{s_m} = L_mZ+O_mZ,\quad m\geq -1,
\end{equation}
where $L_m$ are the Virasoro operators constructed in \cite{dubrovin1999frobenius} and $O_m$ are some differential polynomials, and the flows $\diff{}{s_m}$ satisfy the Virasoro commutation relations
\[
\fk{\diff{}{s_k}}{\diff{}{ s_l}} = (l-k)\diff{}{ s_{k+l}},\quad k,l\geq -1.
\]
\end{Main}

Let us briefly explain the basic idea for proving this theorem. Consider the following system of evolutionary PDEs with time variable $t$ and spacial variable $x$:
\begin{equation}
\label{1-1}
\diff{u^i}{t} = A^i_j(u)u^j_x+\qe\kk{B^i_j(u)u^j_{xx}+C^i_{jk}(u)u^j_xu^k_x}+\cdots,\quad i = 1,\cdots, n,
\end{equation}
We assume that this system is bihamiltonian with respect to the bihamiltonian structure $(P_0,P_1)$ whose leading term is semisimple. We can associate with it a super extension by introducing odd unknown functions $\qth_i$ for $i = 1,\cdots,n$, and by adding odd flows $\diff{}{\qt_0}$ and $\diff{}{\qt_1}$ which correspond respectively to the Hamiltonian structure $P_0$ and $P_1$ (see \cite{liu2020super} and Sect.\,\ref{AV} given below for details). Thus the super extension of \eqref{1-1} consists of the flows $\diff{}{t}$, $\diff{}{\qt_0}$ and $\diff{}{\qt_1}$ for the unknown functions $u^i$ and $\qth_i$, and the fact that the system \eqref{1-1} is bihamiltonian with respect to $(P_0,P_1)$ is equivalent to the following commutation relations:
\[
\fk{\diff{}{t}}{\diff{}{\qt_i}} = 0,\quad \fk{\diff{}{\qt_i}}{\diff{}{\qt_j}} = 0,\quad i,j = 0,1.
\]
\begin{Ex}
Consider the following Korteweg-de Vries (KdV) equation:
\[
\diff{u}{t} = uu_x+\frac{\qe^2}{12}u_{xxx}.
\]
It admits a bihamiltonian structure given by the following Poisson brackets:
\begin{align*}
\{u(x),u(y)\}_0 &= \qd'(x-y),\\\{u(x),u(y)\}_1 &=u(x) \qd'(x-y)+\frac{u_x}{2}\qd(x-y)+\frac{\qe^2}{8}\qd'''(x-y).
\end{align*}
We introduce an odd unknown function $\qth$, and construct the following super extension of the KdV equation:
\begin{align}
\label{DI}
\diff{u}{t} &= uu_x+\frac{\qe^2}{12}u_{xxx},\quad \diff{\qth}{t} = u\qth_x+\frac{\qe^2}{12}\qth_{xxx},\\
\diff{u}{\qt_0} &= \qth_x,\quad \diff{u}{\qt_1} = v\qth_x+\frac{1}{2}v_x\qth+\frac{\qe^2}{8}\qth_{xxx},\\
\diff{\qth}{\qt_0}&=0,\quad \diff{\qth}{\qt_1}=\frac 12 \qth\qth_x.
\end{align}
It is easy to check directly that the flows in the extended system mutually commute.
\end{Ex}
\begin{Rem}
The flow \eqref{DI} also appeared in \cite{becker1993non}.
\end{Rem}

According to the theory of bihamiltonian cohomology given in \cite{DLZ-1}, we know that if there is another system of evolutionary PDEs given by the flow
$\diff{}{\hat t}$ satisfying the commutation relation
\begin{equation}
\label{1-2}
\fk{\diff{}{\hat t}}{\diff{}{\qt_0}} = \fk{\diff{}{\hat t}}{\diff{}{\qt_1}} =0,
\end{equation}
then it gives a symmetry of the system \eqref{1-1}, i.e.,
\[
\fk{\diff{}{t}}{\diff{}{\hat t}} = 0.
\]
We note that the above results is for a flow $\diff{}{\hat t}$ given by differential polynomials. In order to use the above results to consider symmetries of more general forms such as Virasoro symmetries, we introduce the notion of super tau-cover and develop the theory of variational bihamiltonian cohomology \cite{liu2021variational}. Then we are able to consider the commutation relation between the Virasoro symmetries $\diff{}{s_m}$ and the odd flows $\diff{}{\qt_0}$ and $\diff{}{\qt_1}$. By applying the results of variational bihamiltonian cohomology proved in \cite{liu2021variational}, we manage to prove the main theorem.

This paper is organized as follows. In Sect.\,\ref{sec2}, we construct the super tau-cover for a given tau-symmetric bihamiltonian deformation of the Principal Hierarchy associated with a semisimple Frobenius manifold. This construction builds a bridge which relates the Virasoro symmetries to the bihamiltonian structures. In Sect.\,\ref{sec3}, we explain how the deformation problem of the Virasoro symmetries can be solved via the theory of the variational bihamiltonian cohomology. Sect.\,\ref{sec4} is devoted to the proof of the main theorem. Finally some concluding remarks are made in Sect.\,\ref{sec5}.

\section{Super tau-covers of bihamiltonian integrable hierarchies}
\label{sec2}
\subsection{Bihamiltonian structures on the infinite jet space}
Let us start by recalling the basic constructions of bihamiltonian structures as local functionals on infinite jet spaces. One may refer to \cite{liu2018lecture} for a detailed introduction to this topic. 

Let $M$ be a smooth manifold of dimension $n$ and $\hat M$ be the super manifold of dimension $(n|n)$ obtained by reversing the parity of the fibers of the cotangent bundle of $M$. In another word, if we choose a local canonical coordinate system $(u^1,\cdots,u^n;\qth_1,\cdots,\qth_n)$ on $T^*M$, then $\hat M$ can be described locally by the same chart while regarding the fiber coordinates as odd variables:
\[
\qth_i\qth_j+\qth_j\qth_i = 0,\quad i,j=1,\cdots,n.
\]
We call that the odd coordinates $\qth_i$ are dual to $u^i$. The transition functions between two local trivializations $(u^1,\cdots,u^n;\qth_1,\cdots,\qth_n)$ and $(w^1,\cdots,w^n;\phi_1,\cdots,\phi_n)$ are given by the same formula as those of the cotangent bundle:
\begin{equation*}
\phi_\qa = \diff{u^\qb}{w^\qa}\qth_\qb,\quad \qa = 1,\cdots,n.
\end{equation*}
Here and henceforth, summation of the repeated upper and lower Greek indices is assumed.

Denote by $J^\infty(\hat M)$ the infinite jet bundle of $\hat M$. It is a fiber bundle over $\hat M$ with fiber $\mathbb R^\infty$. If we choose a local chart $(u^1,\cdots,u^n;\qth_1,\cdots,\qth_n)$ of $\hat M$, a trivialization can be realized by choosing the fiber coordinate being $(u^{\qa,s};\qth_\qa^s)$ for $\qa = 1,\cdots,n$ and $s\geq 1$. The transition functions between different charts are given by the chain rule
\[
w^{\qa,1} = \diff{w^\qa}{u^\qb}u^{\qb,1},\quad w^{\qa,2} = \diff{w^\qa}{u^\qb}u^{\qb,2}+\frac{\qp^2w^\qa}{\qp u^\qb \qp u^\qg}u^{\qb,1} u^{\qg,1},\quad\cdots
\]
\[
\phi_\qa^1 = \diff{u^\qb}{w^\qa}\qth_\qb^1+\frac{\qp^2 u^\qb}{\qp w^\qg\qp w^\qa}\diff{w^\qg}{u^\ql}u^{\ql,1}\qth_\qb,\quad \cdots
\]
Denote by $\hm A$ the ring of differential polynomials, locally it is given by
\[\hat{\mathcal A} = C^\infty(u)[[u^{i,s+1}, \theta^s_i\mid i = 1,\cdots,n;s\geq 0]].\]
It is graded with respect to the super degree $\deg_\qth$ defined by
\[
\deg_{\qth}u^{\qa,s} = 0,\quad \deg_{\qth}\qth_\qa^s = 1;\quad\qa = 1,\cdots,n,\ s\geq 0.
\]
Here and henceforth we use the notation $u^{\qa,0} = u^{\qa}$, $\qth_\qa^0 = \qth_\qa$. The set of homogeneous elements with super degree $p$ is denoted by $\hm A^p$.

Introduce a global vector filed $\qp_x$ on $J^\infty(\hat M)$ which is locally described by
\[
\qp_x = \sum_{s\geq 0}u^{\qa,s+1}\diff{}{u^{\qa,s}}+\qth_\qa^{s+1}\diff{}{\qth_\qa^s}.
\]
Therefore we see that $u^{\qa,s} = \qp_x^su^\qa$ and $\qth_\qa^s = \qp_x^s\qth_\qa$. Hence we can grade the ring $\hm A$ with respected to the differential degree $\deg_{\qp_x}$ defined by
\[
\deg_{\qp_x}u^{\qa,s} = s,\quad\deg_{\qp_x}\qth_\qa^s = s;\quad \qa = 1,\cdots,n,\ s\geq 0.
\]
We use the notation $\hm A_d$ to denote the set of homogeneous elements with differential degree $d$, and $\hm A^p_d = \hm A^p\cap \hm A_d$.

Using the vector field $\qp_x$, one can construct the space $
\hm F$ of local functionals via the quotient $\hm F:=\hm A/\qp_x\hm A$. Since the vector field $\qp_x$ is homogeneous with respect to both the super degree and the differential degree, the quotient space $\hm F$ admits natural gradations induced from $\hm A$ and we will use the notation $\hm F^p$, $\hm F_d$ and $\hm F^p_d$ to denote the subspace of homogeneous elements. For any element $f\in\hm A$, we will use $
\int f\in\hm F$ to denote its image of the natural projection $\hm A\to\hm F$.

For a differential polynomial $f\in \hm A$, one may define the variational derivatives by
\[
\vard{f}{u^\qa} = \sum_{s\geq 0}(-\qp_x)^s\diff{f}{u^{\qa,s}},\quad \vard{f}{\qth_\qa} = \sum_{s\geq 0}(-\qp_x)^s\diff{f}{\qth_\qa^s}.
\]
It is easy to verify that the variational derivatives annihilate the elements in $\qp_x\hm A$, hence the variational derivatives are also well-defined on the quotient space $\hm F$. For any $F\in\hm F$, we have
\[
\vard{F}{u^\qa} = \sum_{s\geq 0}(-\qp_x)^s\diff{f}{u^{\qa,s}},\quad \vard{F}{\qth_\qa} = \sum_{s\geq 0}(-\qp_x)^s\diff{f}{\qth_\qa^s},
\]
with $f\in\hm A$ being an arbitrary lift of $F$ such that $F = \int f$.
With the help of the notion of the variational derivatives, one can define the so-called Schouten-Nijenhuis bracket, which is a bilinear map $[-,-]:\hm F\times\hm F\to\hm F$ given by
\[
\fk{P}{Q} = \int \left(\frac{\qd P}{\qd\qth_\qa}\frac{\qd Q}{\qd u^\qa}+(-1)^p\frac{\qd P}{\qd u^\qa}\frac{\qd Q}{\qd \qth_\qa}\right),\quad P\in \mathcal{\hat F}^p,\, Q\in \mathcal{\hat F}^q.
\]
This bracket satisfies the graded commuting relation
\[\fk{P}{Q} = (-q)^{pq}\fk{Q}{P},\quad P\in \mathcal{\hat F}^p,\, Q\in \mathcal{\hat F}^q,\]
and the graded Jacobi identity:
\[
(-1)^{rp}[[P,Q],R]+(-1)^{pq}[[Q,R],P]+(-1)^{qr}[[R,P],Q]=0,\quad P\in \mathcal{\hat F}^p,\, Q\in \mathcal{\hat F}^q,\, R\in \mathcal{\hat F}^r.
\]
Any local functional $P\in\hm F^p$ gives rise to a graded derivation $D_P\in \der^{p-1}$ by 
\begin{equation}
\label{2-2}
D_P = \sum_{s\geq 0}\qp_x^s\kk{\vard{P}{\qth_\qa}}\diff{}{u^{\qa,s}}+(-1)^p\qp_x^s\kk{\vard{P}{u^\qa}}\diff{}{\qth^{s}_\qa}.
\end{equation}
Here the space $\der^p\ (p\in\mathbb Z)$ is the space of linear maps $D:\hm A^q\to \hm A^{q+p}$ satisfying the graded Leibniz rule:
\[
D(fg) = D(f)\,g+(-1)^{kp}f\,D(g),\quad f\in\hm A^k,\ g\in\hm A.
\]
and we denote $\der = \oplus_{p\in\mathbb Z}\der^p$, which is a graded Lie algebra with the graded commutator
\[
[D_1,D_2] = D_1\comp D_2-(-1)^{kl}D_2\comp D_1,\quad D_1\in\der^k,\ D_2\in\der^l.
\]
$\der$ is also graded by the differential degree
\[
\der_d = \{D\in\der\mid D(\hm A_k)\subseteq \hm A_{k+d}\}.
\]
For $P\in\hm F^p$ and $Q\in\hm F^q$, the derivation defined in \eqref{2-2} satisfies the identity
\begin{equation}
\label{dc}
(-1)^{p-1}D_{[P,Q]} = \fk{D_P}{D_Q}.
\end{equation}
The following identities are useful:
\begin{equation}
\label{a2-3}
\vard{}{u^\qa}[P,Q] = D_P\kk{\vard{Q}{u^\qa}}+(-1)^{pq}D_Q\kk{\vard{P}{u^\qa}};
\end{equation}
\begin{equation}
\label{a2-4}
(-1)^{p-1}\vard{}{\qth_\qa}[P,Q] = D_P\kk{\vard{Q}{\qth_\qa}}-(-1)^{(p-1)(q-1)}D_Q\kk{\vard{P}{\qth_\qa}}.
\end{equation}

A Hamiltonian structure is defined as a local functional $P\in\hm F^2$ such that $[P,P] = 0$. We can associate a matrix valued differential operator $\mathcal P = (\mathcal P^{\qa\qb})$ with $\mathcal P^{\qa\qb} = \sum_{s\geq 0}\mathcal P^{\qa\qb}_s\qp_x^s$ to a bivector $P\in\hm F^2$. Here $\mathcal P^{\qa\qb}_s\in\hm A$ is defined by 
\[
\vard{P}{\qth_\qa} = \sum_{s\geq 0}\mathcal P^{\qa\qb}_s\qth_\qb^s,\quad \qa = 1,\cdots,n.
\]
If $P$ is a Hamiltonian structure, then we call $\mathcal P$ the Hamiltonian operator of $P$.
\begin{Th}[\cite{dubrovin1996hamiltonian}]
\label{hydro-flat}
Let $P\in\hm F^2_1$  whose associated differential operator is given by
\[
\mathcal P^{\qa\qb} = g^{\qa\qb}(u)\qp_x+\Qg^{\qa\qb}_\qg(u) u^{\qg,1},\quad \det(g^{\qa\qb})\neq 0.
\]
Then $P$ is a Hamiltonian structure if and only if $g = (g_{\qa\qb}) = (g^{\qa\qb})^{-1}$ defines a flat (pseudo-)Rimeannian metric on $M$ and the Christoffel symbols of the Levi-Civita connection of $g$ are given by $\Qg_{\qa\qb}^\qg = -g_{\qa\ql}\Qg^{\ql\qg}_\qb$.
\end{Th}
The Hamiltonian structure $P$ satisfying the conditions of the above theorem is called of hydrodynamic type. It follows from the above theorem that for each  Hamiltonian structure $P$ of hydrodynamic type, there exists a local coordinate system $(v^\qa;\qs_\qa)$ on $\hat M$ such that 
\[
P = \frac 12\int\eta^{\qa\qb}\qs_\qa\qs_\qb^1,
\]
where $\eta^{\qa\qb}$ is a constant non-degenerate matrix. The coordinates $v^\qa$ and $\qs_\qa$ are called  flat coordinates of $P$.

A bihamiltonian structure $(P_0,P_1)$ is a pair of Hamiltonian structure satisfying an additional compatibility condition $[P_0,P_1] = 0$. Assume that the bihamiltonian structure is of hydrodynamic type, then according to Theorem \ref{hydro-flat}, we have two flat contravariant metrics $g^{\qa\qb}_0$ and $g^{\qa\qb}_1$. We say that this bihamiltonian structure is semisimple if the roots of the characteristic equation 
\[
\det\kk{g_1^{\qa\qb}-\ql g_0^{\qa\qb}} = 0
\]
are distinct and not constant. In this case, the roots $\ql^1(u),\cdots,\ql^n(u)$ can serve as local coordinates of $M$ and it is called the canonical coordinate of the semisimple bihamiltonian structure. It is proved in \cite{ferapontov2001compatible} that in canonical coordinate,
\[
P_0 = \frac 12\int \sum_{i,j=1}^n \kk{\qd_{i,j}f^i(\ql)\qth_i\qth_i^1+ A^{ij}\qth_i\qth_j},\quad P_1 = \frac 12\int  \sum_{i,j=1}^n\kk{\qd_{i,j} g^i(\ql)\qth_i\qth_i^1+ B^{ij}\qth_i\qth_j},
\]
where $f^i$ are non-vanishing functions, $g^i = \ql^if^i$ and the functions $A^{ij}$ and $B^{ij}$ are given by
\begin{equation}
\label{DQ}
A^{ij} = \frac 12\left(\frac{f^i}{f^j}\diff{f^j}{\ql^i}\ql^{j,1}-\frac{f^j}{f^i}\diff{f^i}{\ql^j}\ql^{i,1}\right),\quad B^{ij} = \frac 12\left(\frac{g^i}{f^j}\diff{f^j}{\ql^i}\ql^{j,1}-\frac{g^j}{f^i}\diff{f^i}{\ql^j}\ql^{i,1}\right).
\end{equation}
Here by abusing the notation, we still use $\qth_i$ to denote the fiber coordinates of $\hat M$ dual to $\ql^i$. We also call $\ql^i$ and $\qth_i$ the canonical coordinates of $(P_0,P_1)$.

From now on, for a semisimple hydrodynamic bihamiltonian structure $(P_0,P_1)$, we will use $(v^\qa;\qs_\qa)$ to denote the flat coordinate of $P_0$ such that
\[
P_0 = \frac 12\int\eta^{\qa\qb}\qs_\qa\qs_\qb^1,\quad P_1 = \frac 12\int g^{\qa\qb}(v)\qs_\qa\qs_\qb^1+\Qg^{\qa\qb}_\qg(v)v^{\qg,1}\qs_\qa\qs_\qb,
\]
and we will use $(u^i;\qth_i)$ to denote the canonical coordinate for $(P_0,P_1)$ such that 
\begin{equation}
\label{CW}
P_0 = \frac 12\int \sum_{i,j=1}^n\kk{\qd_{i,j} f^i(u)\qth_i\qth_i^1+ A^{ij}\qth_i\qth_j},\quad P_1 = \frac 12\int \sum_{i,j=1}^n\kk{\qd_{i,j} u^if^i(u)\qth_i\qth_i^1+ B^{ij}\qth_i\qth_j}.
\end{equation}
Note that we will not adopt Einstein summation for Latin indices.

In terms the notations introduced above, a system of evolutionary PDEs
\[
\diff{u^\qa}{t} = X^\qa,\quad X^\qa\in\hm A^0
\]
can be represented as a local functional $X = \int X^\qa\qth_{\qa}$, and it is said to be bihamiltonian if there exists a bihamiltonian structure $(P_0,P_1)$ and two Hamiltonians $G,H\in\hm F^0$ such that
\[
X = -[G,P_0] =-[H,P_1]. 
\]
\begin{Ex}
The KdV equation
\begin{equation}
\label{DJ}
\diff{u}{t} = uu_x+\frac{\qe^2}{12}u_{xxx}
\end{equation}
can be represented by $X = \int (uu_x+\frac{\qe^2}{12}u_{xxx})\qth$. Its bihamiltonian structure is given by
\[
P_0 = \frac 12\int \qth\qth_x,\quad P_1 = \frac 12\int u\qth\qth^1+\frac{\qe^2}{8}\qth\qth^3.
\]
The two Hamiltonians with respect to the bihamiltonian structure are given by
\[
X = -\fk{\int \frac{u^3}{6}-\frac{\qe^2}{24}u_x^2}{P_0} = -\fk{\int \frac{u^2}{3}}{P_1}.
\]
\end{Ex}

\subsection{Frobenius manifolds and super tau-covers of the Principal Hierarchies}
In this section, we will recall some basic facts of Frobenius manifolds and the construction of the associated Principal Hierarchies following the work \cite{dubrovin1993integrable,dubrovin1996geometry,dubrovin1999painleve,dubrovin2001normal}. Then we recall the construction of the super tau-covers of the Principal Hierarchies given in \cite{liu2020super}.

The notion of Frobenius manifold is a geometric description of genus zero 2D topological field theories. An $n$-dimensional Frobenius manifold $M$ can be locally described by a solution $F(v^1,\cdots,v^n)$ of the following Witten-Dijkgraaf-Verlinde-Verlinde (WDVV) associativity equation \cite{dijkgraaf1991topological,witten1990structure}:
\begin{equation}
\label{WDVV}
\qp_\qa\qp_\qb\qp_\ql F\eta^{\ql\mu}\qp_\mu\qp_\qg\qp_\qd F = \qp_\qd\qp_\qb\qp_\ql F\eta^{\ql\mu}\qp_\mu\qp_\qg\qp_\qa F.
\end{equation}
Here $\qp_\qa = \diff{}{v^\qa}$ and we require $(\eta_{\qa\qb}) := (\qp_1\qp_\qa\qp_\qb F)$ is a constant non-degenerate matrix with inverse $(\eta^{\qa\qb})$. The function $F(v^1,\cdots,v^n)$ is called the potential of $M$ and it defines a Frobenius algebra structure on $TM$:
\[
\langle \qp_\qa,\qp_\qb\rangle = \eta_{\qa\qb},\quad \qp_\qa\cdot\qp_\qb = c^{\qg}_{\qa\qb}\qp_\qg,
\]
where the functions $c^{\qg}_{\qa\qb}$ are defined by
\[
\quad c^{\qg}_{\qa\qb} = \eta^{\qg\ql}c_{\ql\qa\qb},\quad c_{\ql\qa\qb} = \qp_\ql\qp_\qa\qp_\qb F.
\]
The potential $F$ is required to be quasi-homogeneous in the sense that there exists a vector field
\[
E = \sum_{\qa=1}^n \left(\left(1-\frac{d}{2}-\mu_\qa\right)v^\qa+r_\qa\right)\qp_\qa,
\]
called the Euler vector field, such that
\[ E(F)=(3-d)F+\frac12 A_{\qa\qb} v^\qa v^\qb+B_\qa v^\qa+C.\]
Here the diagonal matrix $\mu = \mathrm{diag}(\mu_1,\cdots,\mu_n)$ is part of the isomonodromy data of $M$ and it is assumed that $\mu_1 = -d/2$ and $r_1 = 0$. The matrix $\mu$ satisfies the property
\begin{equation}
\label{2-3}
(\mu_\qa+\mu_\qb)\eta_{\qa\qb} = 0,\quad\forall\, \qa,\qb.
\end{equation}

An important property of Frobenius manifolds is that the affine connection
\[
\tilde\nabla_X Y = \nabla_X Y+zX\cdot Y,\quad \forall\, X,Y\in\Qg(TM),\quad z\in\mathbb C
\] 
is flat for arbitrary $z$, here $\nabla$ is the Levi-Civita connection of the flat metric $(\eta_{\qa\qb})$. It can be extended to be a flat connection on $M\times \mathbb C^*$ by viewing $z$ as the coordinate on $\mathbb C^*$ and defining
\[
\tilde\nabla_{\qp_z}X = \qp_z X+E\cdot X-\frac 1z \mu X,\quad \tilde\nabla_{\qp_z}\qp_z = \tilde\nabla_{X}\qp_z = 0.
\]
The connection $\tilde\nabla$ is called the deformed flat connection or the Dubrovin connection. For such a flat connection, one can find a system of flat coordinates of the form
\[
(\tilde v^1(v,z),\cdots,\tilde v^n(v,z)) = (h_1(v,z),\cdots,h_n(v,z))z^\mu z^R.
\]
Here $R$ is a constant matrix. The constant matrices $\eta$, $\mu$ and $R$ form the monodromy data of $M$ at $z=0$. The matrix $R$ can be decomposed into a finite sum $R = R_1+\cdots+R_m$, and they satisfy the condition
\begin{equation}
\label{2-4}
\fk{\mu}{R_k} = kR_k,\quad \eta_{\qa\qg}(R_k)^\qg_\qb =(-1)^{k+1}\eta_{\qb\qg}(R_k)^\qg_\qa.
\end{equation} 
The functions $h_\qa(v,z)$ are analytic at $z=0$ and has the expansion $h_\qa(v,z) = \sum_{p\geq 0}h_{\qa,p}(v)z^p$. The coefficients $h_{\qa,p}$ satisfy the following recursion condition:
\[
h_{\qa,0} = \eta_{\qa\qb}v^\qb,\quad\qp_\qb\qp_\qg h_{\qa,p+1} = c^\ql_{\qb\qg}\qp_\ql h_{\qa,p},\quad p\geq 0.
\]
They also satisfy the following quasi-homogeneous condition
\[
E(\qp_\qb h_{\qa,p}) = (p+\mu_\qa+\mu_\qb)\qp_\qb h_{\qa,p}+\sum_{k=1}^p(R_k)^\qg_\qa\qp_\qb h_{\qg,p-k},
\]
and the following normalization condition
\[
\langle\nabla h_\qa(v,z),\nabla h_\qb(v,-z)\rangle = \eta_{\qa\qb}.
\]
A choice of the functions $h_{\qa,p}$ satisfying all the conditions above is called a calibration of $M$ and a Frobenius manifold $M$ is called calibrated if such a choice is fixed. In what follows, we assume that all the Frobenius manifolds are calibrated.

The Principal Hierarchy associated with a Frobenius manifold $M$ is a bihamiltonian integrable hierarchy of hydrodynamic type. Denote by $\qs_\qa$ the odd variables dual to the flat coordinates $v^\qa$, then the Principal Hierarchy can be described by the local functionals $X_{\qa,p}\in\hm F^1$ of the form
\begin{equation}
\label{DK}
X_{\qa,p} = \int \eta^{\ql\qg}\qp_x(\qp_\qg h_{\qa,p+1})\qs_\ql,\quad \qa = 1,\cdots,n,\quad p\geq 0,
\end{equation}
or equivalently, we can represent the hierarchy as follows:
\[
\diff{v^\ql}{t^{\qa,p}} = \eta^{\ql\qg}\qp_x(\qp_\qg h_{\qa,p+1}).
\]
Define two local functionals
\begin{equation}
\label{biham-frob}
P_0 = \frac 12\int\eta^{\qa\qb}\qs_\qa\qs_\qb^1,\quad P_1 = \frac 12\int g^{\qa\qb}\qs_\qa\qs_\qb^1+\Qg^{\qa\qb}_\qg v^{\qg,1}\qs_\qa\qs_\qb,
\end{equation}
where the functions $g^{\qa\qb}$ and $\Qg^{\qa\qb}_\qg$ are given by
\[
g^{\qa\qb} = E^\qe c^{\qa\qb}_\qe,\quad \Qg^{\qa\qb}_\qg = \left(\frac{1}{2}-\mu_\qb\right)c^{\qa\qb}_\qg;\quad c^{\qa\qb}_\qg = \eta^{\qa\ql}c^\qb_{\ql\qg}.
\]
Then we have the following theorem.
\begin{Th}[\cite{dubrovin1996geometry}]
\label{AG}
Let $M$ be a Frobenius manifold, then
\begin{enumerate}
\item The local functionals $P_0$, $P_1$ defined in \eqref{biham-frob} form a bihamiltonian structure which is exact in the sense that
\[
P_0 = [Z,P_1],\quad Z = \int\qs_{1}.
\]
\item The Principal Hierarchy $X_{\qa,p}$ associated with $M$ is bihamiltonian with respect to the bihamiltonian structure $(P_0,P_1)$ and 
\[
X_{\qa,p} = -\fk{H_{\qa,p}}{P_0},\quad H_{\qa,p} = \int h_{\qa,p+1}.
\]
\item The following bihamiltonian recursion relation holds true:
\[
[H_{\qa,p-1},P_1]=\left(p+\frac12+\mu_\qa\right)[H_{\qa,p},P_0]+\sum_{k=1}^p \left(R_k\right)^\qg_\qa [H_{\qg, p-k},P_0],\quad p\geq 0.
\]
\end{enumerate}
\end{Th}

Another important property satisfied by the Principal Hierarchy is that it is tau-symmetric. Let functions $\Qo_{\qa,p;\qb,q}$ for $\qa,\qb = 1,\cdots,n$ and $p,q\geq 0$ be defined by the generating function
\[
\sum_{p\geq 0,q\geq 0}\Qo_{\qa,p;\qb,q}(v)z_1^pz_2^q = \frac{\langle\nabla h_\qa(v,z_1),\nabla h_\qb(v,z_2)\rangle - \eta_{\qa\qb}}{z_1+z_2}.
\]
They have the following properties \cite{dubrovin1996geometry}:
\[
\Qo_{\qa,p;1,0} = h_{\qa,p},\quad \Qo_{\qa,p;\qb,0} = \qp_\qb h_{\qa,p+1};
\]
\[
\Qo_{\qa,p;\qb,q} = \Qo_{\qb,q;\qa,p},\quad \diff{\Qo_{\qa,p;\qb,q}}{t^{\ql,k}}= \diff{\Qo_{\ql,k;\qb,q}}{t^{\qa,p}}.
\]
It follows from these identities that one can extend the Principal Hierarchy by introducing another family of unknown functions $f_{\qa,p}$ satisfying the following equations:
\begin{equation}
\label{g0tau}
\diff{f_{\qa,p}}{t^{\qb,q}} = \Qo_{\qa,p;\qb,q}, \quad \diff{v^\qa}{t^{\qb,q}} = \eta^{\qa\ql}\qp_x\Qo_{\ql,0;\qb,q}.
\end{equation}
The system \eqref{g0tau} is called the tau-cover of the Principal Hierarchy. 

In order to study the relations between the bihamiltonian structures and the Virasoro symmetries, the notion of super tau-covers was introduced in \cite{liu2020super}. Let us briefly recall the construction of the super tau-covers since it provides the main motivation of our work presented in the next subsection. We first introduce a family of odd unknown functions $\qs_{\qa,k}^s$ for $s,k\geq 0$ with $\qs_{\qa,0}^s = \qs^s_\qa$. In what follows we will also use $\qs_{\qa,k}$ to denote $\qs_{\qa,k}^0$. We extend the action of $\qp_x$ to include these new odd variables as follows:
\[
\qp_x = \sum_{s\geq 0}v^{\qa,s+1}\diff{}{v^{\qa,s}}+\sum_{k,s\geq 0}\qs_{\qa,k}^{s+1}\diff{}{\qs_{\qa,k}^s}.
\] 
These odd variables are required to satisfy the following bihamiltonian recursion relation:
\begin{equation}
\label{2-8}
\eta^{\qa\qb}\qs_{\qb,k+1}^1 = g^{\qa\qb}\qs_{\qb,k}^1+\Qg_\qg^{\qa\qb}v^{\qg,1}\qs_{\qb,k},\quad \qa=1,\dots, n;\, k\geq 0.
\end{equation}
We also introduce a family of odd flows $\diff{}{\qt_m}$ for $m\geq 0$. The first two odd flows are determined by the bihamiltonian structure $(P_0,P_1)$:
\[
\diff{v^\qa}{\qt_i} = \vard{P_i}{\qs_\qa},\quad \diff{\qs_\qa}{\qt_i} = \vard{P_i}{v^\qa},\quad i = 0,1.
\]
Note that $\diff{}{\qt_i} = D_{P_i}$ for $i  =0,1$, where $D_{P_i}$ is defined by \eqref{2-2}. The actions of $\diff{}{\qt_i}$ can be extended to all the other odd variables $\qs_{\qa,k}$ such that the flows $\diff{}{\qt_i}$ are compatible with the recursion relations \eqref{2-8}.
Furthermore we can define infinitely many odd flows $\diff{}{\qt_m}$ for $m\geq 2$ which can be viewed as flows corresponding to certain non-local Hamiltonian structures.

We have the following theorem.
\begin{Th}[\cite{liu2020super}]
\label{g0super1}
We have the following mutually commuting flows associated with any given Frobenius manifold $M$:
\begin{align*}
&\diff{v^\qa}{t^{\qb,p}} = \eta^{\qa\qg}{(\qp_\ql\qp_\qg h_{\qb,p+1})v^{\ql,1}},\quad \diff{\qs_{\qa,k}}{t^{\qb,p}} = \eta^{\qg\qe}(\qp_\qa\qp_\qe h_{\qb,p+1})\qs_{\qg,k}^1,\\
&\diff{v^\qa}{\qt_m} = \eta^{\qa\qb}\qs_{\qb,m}^1,\quad
\diff{\qs_{\qa,k}}{\qt_m} = -\diff{\qs_{\qa,m}}{\qt_k} = \Qg^{\qg\qb}_\qa\sum_{i = 0}^{m-k-1}\qs_{\qb,k+i}\qs_{\qg,m-i-1}^1,\quad 0\le k\leq m.
\end{align*}
where $\qa,\qb=1,\dots,n$, and $m, p\ge 0$. These flows are well-defined in the sense that they are compatible with the recursion relations \eqref{2-8}.
\end{Th}

The system described in Theorem \ref{g0super1} is a super extension of the Principal Hierarchy, since the reduction obtained by setting all the odd variables to be zero yields the original Principal Hierarchy. The super extension of the tau-cover \eqref{g0tau} can be constructed by introducing another family of odd variables $\Phi^{m}_{\qa,p}$ for $p,m\geq 0$ and we call it the super tau-cover of the Principal Hierarchy. It is given in the following theorem.
\begin{Th}[\cite{liu2020super}]
\label{g0super2}
The mutually commuting flows
\begin{align*}
\diff{f_{\qa,p}}{t^{\qb,q}} &= \Qo_{\qa,p;\qb,q},\\
\diff{f_{\qa,p}}{\qt_m} &= \Phi_{\qa,p}^m,\\
\diff{\Phi_{\qa,p}^m}{t^{\qb,q}} &= \diff{\Qo_{\qa,p;\qb,q}}{\qt_m},\\
\diff{\Phi_{\qa,p}^m}{\qt_k} &= \Qd_{\qa,p}^{k,m},
\end{align*}
together with the ones presented in Theorem \ref{g0super1}, give the super tau-cover of the Principal Hierarchy associated with $M$, where $\Qd_{\qa,p}^{k,m}$ are defined by the formula
\begin{equation*}
\Qd_{\qa,p}^{k,m}=-\Qd_{\qa,p}^{m,k}=\eta^{\qg\ql}\qp_\ql h_{\qa,p}\Qg_\qg^{\qd\mu}\left(\sum_{i=0}^{k-m-1}\qs_{\mu,m+i}\qs_{\qd,k-i-1}^1\right),\quad k\ge m.
\end{equation*}
\end{Th}

It was shown in \cite{liu2020super} that the odd variables $\Phi_{\qa,p}^m$ satisfy the recursion relations
\[
-\left(\frac{2p-1}{2}+\mu_\qa\right)\Phi_{\qa,p}^m = \left(\frac 12+\mu_\ql\right)\eta^{\ql\qe}(\qp_\ql h_{\qa,p})\qs_{\qe,m}+\sum_{k=1}^p(R_k)^\xi_\qa\Phi_{\xi,p-k}^m-\Phi_{\qa,p-1}^{m+1}
\]
with the initial condition $\Phi^m_{\qa,0} = \qs_{\qa,m}$. So when the diagonal matrix $\mu$ of a Frobenius manifold $M$ satisfies the condition $\frac{1-2k}{2}\notin\mathrm{Spec}(\mu)$ for any $k = 1,2,\cdots$, all the variables $\Phi^m_{\qa,p}$ are linear combinations of $\qs_{\qe,k}$ with coefficients being smooth functions of $v^1,\cdots,v^n$.

For an arbitrary tau-symmetric bihamiltonian deformation of the Principal Hierarchy associated with a semisimple Frobenius manifold, we are to construct in the remaining part of this section its super extension and super tau-cover by generalizing the constructions given in Theorem \ref{g0super1} and Theorem \ref{g0super2}.

\subsection{Super extensions of bihamiltonian integrable hierarchies}
\label{AV}
In this subsection, we construct a super extension for a given bihamiltonian integrable hierarchy with hydrodynamic limit.

We fix an $n$-dimensional smooth manifold $M$ and a semisimple hydrodynamic bihamiltonian structure $(P_0^{[0]},P_1^{[0]})$ defined on $J^\infty(\hat M)$. Let us choose $(w^\qa;\phi_{\qa})$ as local coordinates on $\hat M$.
Recall that a Miura type transformation is a choice of $n$ differential polynomials $\tilde w^1,\cdots \tilde w^n\in\hm A^0$ such that
\[
\det\kk{\diff{\tilde w^\qa_0}{w^\qb}}\neq 0,
\]
where $\tilde w^\qa_0$ is the differential degree zero component of $\tilde w^\qa$. By defining $\tilde w^{\qa,s} = \qp_x^s\tilde w^\qa$, it is easy to see that we can represent any differential polynomial in $w^{\qa,s}$ by a differential polynomial in $\tilde w^{\qa,s}$. Therefore a Miura type transformation can be viewed as a special type of coordinate transformation on $J^\infty(M)$. The extension of the Miura type transformations on $J^\infty(\hat M)$ is given by the following theorem \cite{liu2011jacobi}.
\begin{Th}[\cite{liu2011jacobi}]
\label{local-miura}
A Miura type transformation induces a change of coordinates from $(w^{\qa,s};\phi_\qa^s)$ to $(\tilde w^{\qa,s};\tilde \phi_\qa^s)$ given by
\[
\phi_\qa^s = \qp_x^s\sum_{t\geq 0}(-\qp_x)^t\kk{\diff{\tilde w^{\qb}}{w^{\qa,t}}\tilde\phi_\qb}.
\]
\end{Th}

Now let $(P_0,P_1)$ be any given deformation of $(P_0^{[0]},P_1^{[0]})$. Denote by $\mathcal P_0$ and $\mathcal P_1$ the Hamiltonian operators of $P_0$ and $P_1$ under the coordinate $(w^{\qa,s};\phi_\qa^s)$. We introduce another family of odd variables $\phi^s_{\qa,m}$ for $m\geq 0$ and extend the vector field $\qp_x$ to the following one:
\[
\qp_x = \sum_{s\geq 0}w^{\qa,s+1}\diff{}{w^{\qa,s}}+\sum_{k,s\geq 0}\phi_{\qa,k}^{s+1}\diff{}{\phi_{\qa,k}^s}.
\]
In what follows we also use the notations $\phi^s_{\qa,0} = \phi^s_\qa$ and $\phi^0_{\qa,m} = \phi_{\qa,m}$. Inspired by  \eqref{2-8}, we require that these new odd variables satisfy the recursion relations
\begin{equation}
\label{2-9}
\mathcal P_0^{\qa\qb}\phi_{\qb,m+1} = \mathcal P_1^{\qa\qb}\phi_{\qb,m},\quad m\geq 0.
\end{equation}
We first show that \eqref{2-9} is well defined in the sense that it is invariant under Miura type transformations.
\begin{Prop}
The Miura type transformation from $(w^{\qa,s};\phi_\qa^s)$ to $(\tilde w^{\qa,s};\tilde \phi_\qa^s)$ induces a transformation for the new odd variables $\phi^s_{\qa,m}$ given by
\begin{equation}
\label{2-10}
\phi_{\qa,m}^s = \qp_x^s\sum_{t\geq 0}(-\qp_x)^t\kk{\diff{\tilde w^{\qb}}{w^{\qa,t}}\tilde\phi_{\qb,m}},\quad m\geq 1
\end{equation}
such that the recursion relation \eqref{2-9} is invariant.
\end{Prop}
\begin{proof}
Denote by $\mathcal P_i$ and $\tilde{\mathcal P_i}$ the Hamiltonian operator of $P_i$ in the coordinates $(w^{\qa,s};\phi_\qa^s)$ and $(\tilde w^{\qa,s};\tilde \phi_\qa^s)$ respectively for $i = 0,1$. Then it is well known that
\[
\tilde{\mathcal P_i}^{\qa\qb} = \sum_{s\geq 0}\diff{\tilde w^\qa}{w^{\ql,s}}\qp_x^s\comp\mathcal P_i^{\ql\qe}\comp\sum_{s\geq 0}(-\qp_x)^s\comp\diff{\tilde w^\qb}{w^{\qe,s}}.
\]
Therefore by using the relation \eqref{2-9}, it is easy to see that:
\begin{align*}
\tilde{\mathcal P_0}^{\qa\qb}\tilde\phi_{\qb,m+1}&=\sum_{s\geq 0}\diff{\tilde w^\qa}{w^{\ql,s}}\qp_x^s\comp\mathcal P_0^{\ql\qe}\comp\sum_{s\geq 0}(-\qp_x)^s\comp\diff{\tilde w^\qb}{w^{\qe,s}}(\tilde\phi_{\qb,m+1})\\
&=\sum_{s\geq 0}\diff{\tilde w^\qa}{w^{\ql,s}}\qp_x^s\comp\mathcal P_0^{\ql\qe}\phi_{\qe,m+1}\\
&=\sum_{s\geq 0}\diff{\tilde w^\qa}{w^{\ql,s}}\qp_x^s\comp\mathcal P_1^{\ql\qe}\phi_{\qe,m}\\
&=\sum_{s\geq 0}\diff{\tilde w^\qa}{w^{\ql,s}}\qp_x^s\comp\mathcal P_1^{\ql\qe}\comp\sum_{s\geq 0}(-\qp_x)^s\comp\diff{\tilde w^\qb}{w^{\qe,s}}(\tilde\phi_{\qb,m})\\
&=\tilde{\mathcal P_1}^{\qa\qb}\tilde\phi_{\qb,m}.
\end{align*}
Thus we see that the recursion relations \eqref{2-9} are preserved under the change of variables \eqref{2-10}. The proposition is proved.
\end{proof}

By using the theory of bihamiltonian cohomology \cite{DLZ-1} for $(P_0^{[0]},P_1^{[0]})$, we can choose a coordinate system $(v^{\qa,s};\qs_\qa^s)$ such that 
\[
P_0 = P_0^{[0]} = \frac 12\int\eta^{\qa\qb}\qs_\qa\qs_\qb^1,
\]
and the $\hm F^2_2$ component of $P_1$ vanishes. From now on, we will always use $(v^{\qa,s};\qs_\qa^s)$ to denote the coordinate system described above. We use the notation $\hm A^+$ to denote the extension of $\hm A$ by including the odd variables $\qs_{\qa,m}^s$ for $m\geq 1$ satisfying the recursion relations
\begin{equation}
\label{2-11}
\eta^{\qa\qb}\qs_{\qb,m+1}^1 = \mathcal P_1^{\qa\qb}\qs_{\qb,m},\quad m\geq 0.
\end{equation}
As before we use the notation $\qs^s_{\qa,0} = \qs^s_\qa$ and $\qs^0_{\qa,m} = \qs_{\qa,m}$. We will still use $\qp_x$ to denote the vector field on $\hm A^+$ given by
\[
\partial_x = \sum_{s\geq 0}v^{\qa,s+1}\diff{}{v^{\qa,s}}+\sum_{s,m\geq 0}\qs_{\qa,m}^{s+1}\diff{}{\qs_{\qa,m}^s}.
\]
For any element $f\in\hm A^+$, we say $f$ is \textit{local} if it can be represented by an element in $\hm A$ and we say $f$ is \textit{non-local} if it is not local. Note that on the space $\hm A^+$, the differential degree is not well-defined but the super degree is still well defined by setting the super degree of $\qs_{\qa,m}^s$ being 1. We will use $\hm A^{+,p}$ to denote the set of homogeneous elements with super degree $p$.
\begin{Ex}
Consider the following bihamiltonian structure of the KdV equation \eqref{DJ}:
\[
\mathcal P_0 = \qp_x,\quad \mathcal P_1 = v\qp_x+\frac 12 v_x+\frac{\qe^2}{8}\qp_x^3.
\]
We introduce odd variables $\qs_m^s$ for $s, m\geq 0$ such that they satisfy the recursion relations
\[
\qs_{m+1}^1 = v\qs_{m}^1+\frac 12 v_x\qs_{m}+\frac{\qe^2}{8}\qs_{m}^3,\quad m\geq 0.
\]
Then the ring $\hm A^+$ is given by the quotient
\[
\hm A^+ = C^\infty(v)[[v^{(s+1)},\qs_m^s\mid m,s\geq 0]]/J,
\]
where $J$ is the differential ideal generated by
\[
v\qs_{m}^1+\frac 12 v_x\qs_{m}+\frac{\qe^2}{8}\qs_{m}^3-\qs_{m+1}^1,\quad m\geq 0.
\]
Then we see that $\qs_1^1$ is local but $\qs_2^1$ is non-local.
\end{Ex}

\begin{Def}
For $k,l\geq 0$, we define the shift operators
\[
T_k:\hm A^1\to\hm A^{+,1},\quad T_{k,l}:\hm A^2\to\hm A^{+,2}
\]
to be the linear operators given by
\[
T_k(f\qs_{\qa,0}^s) = f\qs_{\qa,k}^s,\quad f\in\hm A^0,
\]
\[
T_{k,l} = -T_{l,k};\quad T_{k,l}(f\qs_{\qa,0}^t\qs_{\qb,0}^s) = f\sum_{i=0}^{l-k-1}\qs_{\qa,k+i}^t\qs_{\qb,l-i-1}^s,\quad k\leq l,\ f\in\hm A^0.
\]
In particular, $T_{k,k} = 0$.
\end{Def}
The following lemmas are obvious from the definition.
\begin{Lem}
The shift operators $T_k$ and $T_{k,l}$ commute with $\qp_x$.
\end{Lem}
\begin{Lem}
The shift operators $T_k$ and $T_{k,l}$ are globally defined, i.e. they are invariant under Miura type transformations. 
\end{Lem}
\begin{Ex}
Using the shift operator, the recursion relation \eqref{2-11} can be represented by the following formula
\begin{equation}
\label{2-12}
T_{m+1}\vard{P_0}{\qs_{\qa,0}} = T_m \vard{P_1}{\qs_{\qa,0}},\quad m\geq 0.
\end{equation}
\end{Ex}
\begin{Ex}
\label{ex2-4}
The recursion relation \eqref{2-11} can also be represented by the following formulae:
\[
T_{m+1}\kk{D_{P_0}f} = T_m\kk{D_{P_1}f},\quad m\geq 0,\ f\in\hm A^0,
\]
here $D_{P_i}$ are the derivations defined in \eqref{2-2}. Indeed, when $f = v^\qa$ we recover the relation \eqref{2-12}; for general $f\in\hm A^0$, by definition \eqref{2-2}, we have
\begin{align*}
T_{m+1}\kk{D_{P_0}f} = T_{m+1}\sum_{s\geq 0}\diff{f}{v^{\qa,s}}\qp_x^s\vard{P_0}{\qs_{\qa,0}} = T_{m}\sum_{s\geq 0}\diff{f}{v^{\qa,s}}\qp_x^s\vard{P_1}{\qs_{\qa,0}} = T_m\kk{D_{P_1}f}.
\end{align*}
\end{Ex}

With the help of the shift operator, we can generalize the construction given in the previous subsection. We first introduce the following notation.
\begin{Def}

We define a family of odd derivations $\diff{}{\qt_m}$ on $\hm A^+$ by
\begin{equation}\label{AH}
\diff{v^\qa}{\qt_m} = T_m\vard{P_0}{\qs_{\qa,0}},\quad \diff{\qs_{\qa,k}}{\qt_m} = T_{k,m}\vard{P_1}{v^\qa},\quad \fk{\diff{}{\qt_m}}{\qp_x} = 0.
\end{equation}
\end{Def}
We need to check that that this definition is well-defined, i.e., it is compatible with the recursion relations \eqref{2-12}. 
\begin{Lem}
\label{lem2-8}
The following identity holds true for any $X\in\hm A^1$ and $m,k\geq 0$:
\[
\diff{}{\qt_k}T_m(X) = T_{m,k}\kk{D_{P_1}(X)}-T_{m+1,k}\kk{D_{P_0}(X)}.
\]
\end{Lem}
\begin{proof}
Since all the operators are linear, we may assume $X = f\qs_{\qb,0}^l$ for some $f\in\hm A^0$ and $l\geq 0$. We first assume that $k\geq m$. The case $k=m$ can be easily verified as follows:
\[
\diff{}{\qt_m}(f\qs_{\qb,m}^l) = T_m\kk{D_{P_0}f}\qs_{\qb,m}^l = T_{m,m+1}\kk{D_{P_0}(f\qs_{\qb,0}^l)}.
\]
Here we use the fact that $P_0 = P_0^{[0]}$ and $D_{P_0}\qs_{\qb,0} = 0$. Now we assume $k\geq m+1$, then by using the definition of the shift operators and the recursion relations  \eqref{2-12} we obtain the following identities:
\begin{align*}
\diff{}{\qt_k}(f\qs_{\qb,m}^l) &= T_k\kk{D_{P_0}f}\qs_{\qb,m}^l+T_{m,k}\kk{fD_{P_1}\qs_{\qb,0}^l}\\
&= T_{k-1}\kk{D_{P_1}f}\qs_{\qb,m}^l+T_{m,k}\kk{fD_{P_1}\qs_{\qb,0}^l}\\
&=T_{m,k}\kk{\kk{D_{P_1}f}\qs_{\qb,0}^l}-\sum_{i=0}^{k-m-2}T_{m+i}\kk{D_{P_1}f}\qs_{\qb,k-i-1}^l+T_{m,k}\kk{fD_{P_1}\qs_{\qb,0}^l}\\
&= T_{m,k}\kk{D_{P_1}(f\qs_{\qb,0}^l)}-\sum_{i=0}^{k-m-2}T_{m+i+1}\kk{D_{P_0}f}\qs_{\qb,k-i-1}^l\\
&=T_{m,k}\kk{D_{P_1}(f\qs_{\qb,0}^l)}-T_{m+1,k}\kk{D_{P_0}(f\qs_{\qb,0}^l)}.
\end{align*}
The case $k<m$ is proved in exactly the same way. The lemma is proved.
\end{proof}
\begin{Prop}
The flows $\diff{}{\qt_k}$ are compatible with the recursion relations \eqref{2-12}, i.e.,
\[
\diff{}{\qt_k}T_{m+1}\vard{P_0}{\qs_{\qa,0}} = \diff{}{\qt_k}T_{m}\vard{P_1}{\qs_{\qa,0}},\quad m,k\geq 0.
\]
\end{Prop}
\begin{proof}
Using the fact that
\[
P_0 = P_0^{[0]} = \frac 12\int\eta^{\qa\qb}\qs_\qa\qs_\qb^1,
\]
it is easy to obtain the following identities:
\[
\diff{}{\qt_k}T_{m+1}\vard{P_0}{\qs_{\qa,0}} = \diff{}{\qt_k}\eta^{\qa\qb}\qs_{\qb,m+1}^1 = T_{m+1,k}\kk{\eta^{\qa\qb}\qp_x\vard{P_1}{v^\qb}} = T_{m+1,k}\kk{D_{P_1}\vard{P_0}{\qs_{\qa,0}}} .
\]
Since $[P_0,P_1] = 0$, it follows from the identity \eqref{a2-4} that
\[
\diff{}{\qt_k}T_{m+1}\vard{P_0}{\qs_{\qa,0}} = -T_{m+1,k}\kk{D_{P_0}\vard{P_1}{\qs_{\qa,0}}}.
\]
Thus by using $[P_1,P_1] = 0$ and Lemma \ref{lem2-8} we finish the proof of the proposition.
\end{proof}
\begin{Prop}
The odd flows $\diff{}{\qt_m}$ mutually commute, i.e.,
\[
\fk{\diff{}{\qt_m}}{\diff{}{\qt_k}} = 0,\quad m,\,k\geq 0.
\]
\end{Prop}
\begin{proof}
By the definition of the flows $\diff{}{\qt_m}$, it is easy to see that
\[
\fk{\diff{}{\qt_m}}{\diff{}{\qt_k}}v^\qa = \eta^{\qa\qb}\qp_xT_{k,m}\vard{P_1}{v^\qb}+\eta^{\qa\qb}\qp_xT_{m,k}\vard{P_1}{v^\qb} = 0.
\]
To show the commutation relation
\[
\fk{\diff{}{\qt_l}}{\diff{}{\qt_k}}\qs_{\qa,m} = 0,
\]
it suffices to verify the case $m = 0$ due to the recursion relations \eqref{2-11}. By using the trivial relation
\[\fk{\diff{}{\qt_0}}{\diff{}{\qt_0}}\qs_{\qa,0} = 0,\]
the recursion relations \eqref{2-11}, and by induction on $k$, we arrive at
\[\fk{\diff{}{\qt_0}}{\diff{}{\qt_0}}\qs_{\qa,k} = 0,\quad k\geq 0.\]
This commutation relation is equivalent to 
\[\fk{\diff{}{\qt_0}}{\diff{}{\qt_k}}\qs_{\qa,0} = 0\]
due to the definition of the odd flows. Using induction again we arrive at the identity
\[\fk{\diff{}{\qt_0}}{\diff{}{\qt_k}}\qs_{\qa,l} = 0\]
for any $l\geq 0$. Therefore we have
\[
\diff{}{\qt_k}\diff{\qs_{\qa,0}}{\qt_l} = -\diff{}{\qt_k}\diff{\qs_{\qa,l}}{\qt_0} = \diff{}{\qt_0}\diff{\qs_{\qa,l}}{\qt_k}.
\]
It follows from the definition of the odd flows that the right hand side is anti-symmetric with respect to the indices $k,l$, hence we prove that
\[
\fk{\diff{}{\qt_l}}{\diff{}{\qt_k}}\qs_{\qa,0} = 0.
\]
The proposition is proved.
\end{proof}

Now let $X_i\in\hm F^1$, $i\in I$ be a family of bihamiltonian vector fields with respect to an index set $I$, i.e., each $X_i$ satisfies the equations $[X_i,P_0] = [X_i,P_1] = 0$. Recall that the family $\{X_i\}$ corresponds to a bihamiltonian integrable hierarchy given by
\begin{equation}
\diff{v^\qa}{t_i} = \vard{X_i}{\qs_{\qa,0}},\quad i\in I.
\end{equation}
In what follows we will extend this integrable hierarchy such that it becomes a system of mutually commuting vector fields on $\hm A^+$.
\begin{Def}
For a bihamiltonian vector field $X\in\hm F^1$ we associate it with the following system of PDEs on $\hm A^+$:
\[
\diff{v^\qa}{t_X} = D_Xv^\qa,\quad \diff{\qs_{\qa,m}}{t_X} = T_m D_X\qs_{\qa,0},\quad \fk{\diff{}{t_X}}{\qp_x} = 0.
\]
It is called the super extended flow given by $X$.
\end{Def}

\begin{Prop}
The super extended flow $\diff{}{t_X}$ given by a bihamiltonian vector field $X$ is compatible with the recursion relations \eqref{2-12}, i.e.,
\begin{equation}
\label{2-15}
\diff{}{t_X}T_{m+1}\vard{P_0}{\qs_{\qa,0}} = \diff{}{t_X}T_{m}\vard{P_1}{\qs_{\qa,0}}.
\end{equation}
\end{Prop}
\begin{proof}
By definition of the flow $\diff{}{t_X}$, it is easy to see that
\[
\diff{}{t_X}T_{m+1}\vard{P_0}{\qs_{\qa,0}} = T_{m+1}D_X\vard{P_0}{\qs_{\qa,0}},\quad  \diff{}{t_X}T_{m}\vard{P_1}{\qs_{\qa,0}}= T_{m}D_X\vard{P_1}{\qs_{\qa,0}}.
\]
Using the fact that $[X,P_0] = [X,P_1] = 0$ and the identity \eqref{a2-4}, we see that \eqref{2-15} is equivalent to the following identity
\[
T_{m+1}D_{P_0}\vard{X}{\qs_{\qa,0}} = T_{m}D_{P_1}\vard{X}{\qs_{\qa,0}},
\]
which is equivalent to the recursion relations \eqref{2-12} according to Example \ref{ex2-4}. The proposition is proved.
\end{proof}
\begin{Prop}
Let $X$ and $Y$ be two bihamiltonian vector fields, then their super extended flows commute.
\end{Prop}
\begin{proof}
From the theory of the bihamiltonian cohomology \cite{DLZ-1} we know that $[X,Y] = 0$, hence it follows from \eqref{dc} that
\[
\fk{\diff{}{t_X}}{\diff{}{t_Y}}v^\qa = 0,\quad \fk{\diff{}{t_X}}{\diff{}{t_Y}}\qs_{\qa,0} = 0.
\]
Using the definition of the super extended flow, it is easy to see that
\[
\fk{\diff{}{t_X}}{\diff{}{t_Y}}\qs_{\qa,m} = T_m\fk{\diff{}{t_X}}{\diff{}{t_Y}}\qs_{\qa,0} = 0.
\]
The proposition is proved.
\end{proof}

Now let us prove that the super extended flow given by a bihamiltonian vector field commutes with the odd flows $\diff{}{\qt_m}$.
\begin{Lem}\label{AB}
For any $\mathcal D\in\der^0$ satisfying the condition $[\mathcal D,\qp_x] = 0$, we extend its action to $\hm A^+$ by setting
\[
\mathcal D\qs_{\qa,m} = T_m\mathcal D\qs_{\qa,0},\quad m\geq 0.
\]
Then the following identities hold true:
\[
T_k\comp\kk{\mathcal D|_{\hm A^1}} = \kk{\mathcal D|_{\hm A^{+,1}}}\comp T_k;\quad T_{k,l}\comp\kk{\mathcal D|_{\hm A^2}} = \kk{\mathcal D|_{\hm A^{+,2}}}\comp T_{k,l},\quad k,l\geq 0.
\]
\end{Lem}
\begin{proof}
The first identity is obvious from the definition $\mathcal D\qs_{\qa,k} = T_k\mathcal D\qs_{\qa,0}$. The second one is also easy to verify by using the definition of the shift operator $T_{k,l}$. The lemma is proved.
\end{proof}
\begin{Prop}
The odd flows $\diff{}{\qt_m}$ commute with the super extended flow given by a bihamiltonian vector field $X$.
\end{Prop}
\begin{proof}
Using Lemma \ref{AB} and the definition of the odd flows, it is easy to see that
\[
\diff{}{t_X}\diff{v^\qa}{\qt_m} = \diff{}{t_X}T_m\vard{P_0}{\qs_{\qa,0}} = T_mD_X\kk{\vard{P_0}{\qs_{\qa,0}}},\quad \diff{}{\qt_m}\diff{v^\qa}{t_X} = T_mD_{P_0}\kk{\diff{v^\qa}{t_X}}.
\]
Therefore it follows from $[X,P_0] = 0$ and the identity \eqref{a2-4} that
\[
\fk{\diff{}{\qt_m}}{\diff{}{t_X}}v^\qa = 0.
\]
Similarly by using Lemma \ref{AB} again we have
\[
\diff{}{t_X}\diff{\qs_{\qa,k}}{\qt_m} = \diff{}{t_X}T_{k,m}\vard{P_1}{v^\qa} = T_{k,m}D_X\kk{\vard{P_1}{v^\qa}}.
\]
On the other hand, by using Lemma \ref{lem2-8} we obtain
\[
\diff{}{\qt_m}\diff{\qs_{\qa,k}}{t_X} = -\diff{}{\qt_m}T_k\kk{D_X\kk{\vard{X}{v^\qa}}} = -T_{k,m}\kk{D_{P_1}\kk{\vard{X}{v^\qa}}}.
\]
Here we use the fact that $[P_0,X] = 0$ and $\vard{P_0}{v^\qa} = 0$. Now by using \eqref{a2-3} and $[X,P_1] = 0$, it is easy to conclude that
\[
\fk{\diff{}{\qt_m}}{\diff{}{t_X}}\qs_{\qa,k} = 0.
\]
The proposition is proved.
\end{proof}

Let us summarize the constructions given in this subsection in the following theorem.
\begin{Th}
\label{AC}
Let $(P_0,P_1)$ be a bihamiltonian structure with semisimple hydrodynamic leading terms and $\{X_i\}$ be a family of bihamiltonian vector fields. Then we have the following super extended integrable hierarchy:
\[\diff{v^\qa}{t_i} = D_{X_i}v^\qa,\quad \diff{\qs_{\qa,m}}{t_i} = T_m D_{X_i}\qs_{\qa,0};\]
\[
\diff{v^\qa}{\qt_m} = T_m\vard{P_0}{\qs_{\qa,0}},\quad \diff{\qs_{\qa,k}}{\qt_m} = T_{k,m}\vard{P_1}{v^\qa}.
\]
The flows in this super extended hierarchy mutually commute.
\end{Th}

\subsection{Deformation of the super tau-covers}
In this section, we will construct super tau-covers for tau-symmetric bihamiltonian  deformations of the Principal Hierarchy associates with a semisimple Frobenius manifold. Let us first recall how to construct the deformations of the tau-cover \eqref{g0tau} of the Principal Hierarchy following \cite{dubrovin2018bihamiltonian}.

We fix a semisimple Frobenius manifold $M$ and use $(P_0^{[0]},P_1^{[0]})$ to denote the bihamiltonian structure \eqref{biham-frob}. The two-point functions in the tau-cover \eqref{g0tau} will be denoted by $\Qo^{[0]}_{\qa,p;\qb,q}$ and similarly we denote $h_{\qa,p}^{[0]} = \Qo^{[0]}_{\qa,p;1,0}$.  Let $(P_0,P_1)$ be a deformation of $(P_0^{[0]},P_1^{[0]})$, then it determines a unique deformation of the Principal Hierarchy according to \cite{DLZ-1,liu2013bihamiltonian}. By using the results proved in \cite{falqui2012exact}, we know that a bihamiltonian deformation of the Principal Hierarchy is tau-symmetric if and only if the central invariants of the deformation of the bihamiltonian structure are constants. In such a case, after an appropriate Miura type transformation the Hamiltonian structure $P_0$ can be represented in the form 
\[
P_0 = P_0^{[0]} = \frac 12\int\eta^{\qa\qb}\qs_\qa\qs_\qb^1,
\]
and $P_1$ has no $\hm F^2_2$ components. Moreover, we can also require that the condition of exactness of the bihamiltonian structure is preserved \cite{dubrovin2018bihamiltonian,falqui2012exact}:
\[
P_0 = [Z,P_1],\quad Z = \int \qs_1.
\]
In what follows, we will always assume that $P_0$ and $Z$ take the above forms. We still use the same notation $X_{\qa,p}\in\hm F^1$, as we have already used in \eqref{DK}, to denote the unique deformed flows of the Principal Hierarchy and we will also use $\diff{}{t^{\qa,p}}$ to denote the vector field $D_{X_{\qa,p}}$. Let $H_{\qa,p}\in\hm F^0$ be the unique deformation of the Hamiltonian of the Principal Hierarchy such that 
\begin{equation}
\label{AD}
X_{\qa,p} = -\fk{H_{\qa,p}}{P_0},\quad \qa = 1,\cdots,n,\quad p\geq 0,
\end{equation}
and
\begin{equation}
\label{DL}
H_{\qa,-1}:=\int \eta_{\qa\qb}v^\qb.
\end{equation}
Let us define the following differential polynomials:
\begin{equation}
\label{AE}
h_{\qa,p} = D_ZH_{\qa, p},\quad \qa = 1,\cdots,n,\quad p\geq 0.
\end{equation}
Note we use an index convention that is different from the one used in \cite{dubrovin2018bihamiltonian}.
\begin{Prop}[\cite{dubrovin2018bihamiltonian}]
\label{AF}
We have the following results:
\begin{enumerate}
\item $D_{X_{1,0}} = \qp_x$.
\item The functionals $H_{\qa,p}$ defined in \eqref{AD},  \eqref{DL} and the differential polynomials defined in \eqref{AE} satisfy the relations
\[
H_{\qa,p} = \int h_{\qa,p+1},\quad p\geq -1.
\] 
\end{enumerate}
\end{Prop}
We also hove the following properties for the Hamiltonians and two-point functions of the deformed Principal Hierarchy.
\begin{Prop}
The following bihamiltonian recursion relations hold true:
\begin{equation}
\label{AN}
[H_{\qa,p-1},P_1]=\left(p+\frac12+\mu_\qa\right)[H_{\qa,p},P_0]+\sum_{k=1}^p \left(R_k\right)^\qg_\qa [H_{\qg, p-k},P_0],\quad p\geq 0.
\end{equation}
\end{Prop}
\begin{proof}
Denote the following by $Y_{\qa,p}\in\hm F^1$ for $p\geq 0$:
\[
Y_{\qa,p} = [H_{\qa,p-1},P_1]-\left(p+\frac12+\mu_\qa\right)[H_{\qa,p},P_0]-\sum_{k=1}^p \left(R_k\right)^\qg_\qa [H_{\qg, p-k},P_0].
\]
According to Theorem \ref{AG}, we know that $Y_{\qa,p}\in\hm F^1_{\geq 2}$. Since the flows $X_{\qa,p} = -[H_{\qa,p},P_0]$ is bihamiltonian, we easily conclude that $Y_{\qa,p}$ is a bihamiltonian vector field. Therefore by using the results of bihamiltonian cohomology given in \cite{DLZ-1}, we arrive at $Y_{\qa,p} = 0$.
\end{proof}
\begin{Th}[\cite{dubrovin2018bihamiltonian}]
\label{AI}
There exist differential polynomials $\Qo_{\qa,p;\qb,q}$ such that they are deformations of $\Qo_{\qa,p;\qb,q}^{[0]}$, and satisfy the following properties:
\begin{enumerate}
\item $\qp_x\Qo_{\qa,p;\qb,q} =\diff{h_{\qb,q}}{t^{\qa,p}}$.
\item $\Qo_{\qa,p;\qb,q} = \Qo_{\qb,q;\qa,p}$ and $\Qo_{\qa,p;1,0} = h_{\qa,p}$.
\item $\diff{\Qo_{\qa,p;\qb,q}}{t^{\ql,k}}= \diff{\Qo_{\ql,k;\qb,q}}{t^{\qa,p}}$.
\end{enumerate}
\end{Th}
To construct the tau-cover of the deformed Principal Hierarchy, we introduce the following \textit{normal coordinates} as in \cite{dubrovin2001normal}:
\begin{equation}
\label{BJ}
w_\qa = h_{\qa,0},\quad w^\qa = \eta^{\qa\qb}w_\qb,
\end{equation}
then we see that the differential polynomials $w^\qa$ and $v^\qa$ are related by a Miura type transformation. In particular, it follows from Proposition \ref{AF} that
\begin{equation}
\label{AK}
H_{\qa,-1} =\int w_\qa =  \int \eta_{\qa\qb}v^\qb.
\end{equation}
In terms of the normal coordinates, the tau-cover of the deformed Principal Hierarchy can be represented in the form (c.p.  \eqref{g0tau}) 
\begin{equation}
\label{AW}
\diff{f_{\qa,p}}{t^{\qb,q}} = \Qo_{\qa,p;\qb,q}, \quad \diff{w^\qa}{t^{\qb,q}} = \eta^{\qa\ql}\qp_x\Qo_{\ql,0;\qb,q}.
\end{equation}

Let us proceed to construct its super tau-cover. To this end we introduce odd variables $\Phi_{\qa,p}^m$, as we do for the super tau-cover of the Principal Hierarchy, such that
\begin{equation}
\label{AJ}
\qp_x\Phi_{\qa,p}^m = \diff{h_{\qa,p}}{\qt_m},\quad m\geq 0.
\end{equation}
Here the odd flows $\diff{}{\qt_m}$ are defined by \ref{AH}. By using Theorem \ref{AI} we obtain the following identity:
\[
\diff{}{t^{\qb,q}}\diff{h_{\qa,p}}{\qt_m} = \qp_x\diff{\Qo_{\qa,p;\qb,q}}{\qt_m}.
\]
Therefore we conclude that the following definitions of the evolutions of the odd variables $\Phi_{\qa,p}^m$ along the flows $\diff{}{t^{\qb,q}}$ are compatible with \eqref{AJ}:
\[
\diff{\Phi_{\qa,p}^m}{t^{\qb,q}} = \diff{\Qo_{\qa,p;\qb,q}}{\qt_m}.
\]
To define the evolutions of $\Phi_{\qa,p}^m$ along the odd flows $\diff{}{\qt_k}$, we need the following lemma.
\begin{Lem}
There exist differential polynomials $F_{\qa,p}\in\hm A^2$ such that
\[
\diff{}{\qt_k}\diff{h_{\qa,p}}{\qt_m} = T_{m,k}\qp_xF_{\qa,p},\quad \qa = 1,\cdots,n;\ p,m,k\geq 0.
\]
\end{Lem}
\begin{proof}
By using the definition \eqref{2-2} of $D_{P_i}$ and the fact that the vector fields $X_{\qa,p}$ are bihamiltonian, it is easy to see that
\[
\int D_{P_1}D_{P_0}h_{\qa,p} = \int[P_1,[P_0,H_{\qa,p-1}]] = 0.
\]
Therefore there exists $F_{\qa,p}\in\hm A^2$ such that
\[
\diff{}{\qt_1}\diff{h_{\qa,p}}{\qt_0} = \qp_xF_{\qa,p}.	
\]
Now from Lemma \ref{lem2-8} it follows that
\[
\diff{}{\qt_k}\diff{h_{\qa,p}}{\qt_m} = T_{m,k}\qp_xF_{\qa,p}.
\]
The lemma is proved.
\end{proof}

Now we are ready to construct the following super tau-cover of the deformed Principal Hierarchy:
\begin{Th}
Let $M$ be a semisimple Frobenius manifold and $(P_0,P_1)$ be a deformation of the bihamiltonian structure \eqref{biham-frob} with constant central invariants. Then we have the following  deformation of the super tau-cover of the Principal Hierarchy:
\begin{align*}
\diff{f_{\qa,p}}{t^{\qb,q}} &= \Qo_{\qa,p;\qb,q},\\
\diff{f_{\qa,p}}{\qt_m} &= \Phi_{\qa,p}^m,\\
\diff{\Phi_{\qa,p}^m}{t^{\qb,q}} &= \diff{\Qo_{\qa,p;\qb,q}}{\qt_m},\\
\diff{\Phi_{\qa,p}^m}{\qt_k} &= T_{m,k}F_{\qa,p}.
\end{align*}
\end{Th}

Recall that when the diagonal matrix $\mu$ of a Frobenius manifold $M$ satisfies the condition $\frac{1-2k}{2}\notin\mathrm{Spec}(\mu)$ for any $k = 1,2,\cdots$, the odd variables $\Phi_{\qa,p}^m$ for the super tau-cover of the Principal Hierarchy are redundant since they can be represented by elements in $\hm A^+$. We want to show that the deformed super tau-cover have the same property which will play an important role in our consideration of the deformation of the Virasoro symmetries.

We start with the definition of generalized shift operators.
\begin{Def}
We define the shift operators $\hat T_k$ for $k\geq 0$ to be linear operators $\hm A^{+,1}\to\hm A^{+,1}$ such that
\[
\hat T_k(f\qs_{\qa,m}^l) = f\qs_{\qa,m+k}^l,\quad f\in\hm A^0,\ m\geq 0.
\]
\end{Def}
The following lemma is easy to verify.
\begin{Lem}
The operators $\hat T_k$ commute with $\qp_x$ and are compatible with the recursion relations \eqref{2-11}.
\end{Lem}
\begin{Lem}
\label{AL}
If $\Phi_{\qa,p}^0$ can be represented by an element in $\hm A^{+,1}$, then so does $\Phi_{\qa,p}^m$ and $\Phi_{\qa,p}^m = \hat T_m\Phi_{\qa,p}^0$ for $m\geq 1$.
\end{Lem}
\begin{proof}
By using the relation \eqref{AJ}, it is easy to see that
\[
\qp_x\hat T_m\Phi_{\qa,p}^0 = \hat T_m \diff{h_{\qa,p}}{\qt_0} = \diff{h_{\qa,p}}{\qt_m}. 
\]
The lemma is proved.
\end{proof}
We will use the notation $\Phi_{\qa,p}^m\in\hm A^+$ to mean that $\Phi_{\qa,p}^m$ can be represented by an element in $\hm A^+$.
\begin{Prop}
\label{BM}
We have $\Phi_{\qa,0}^m\in\hm A^{+}$ and
\begin{equation}
\label{DM}
\kk{\prod_{k=1}^p\kk{k-\frac 12+\mu_\qa}}\Phi_{\qa,p}^m\in\hm A^+,\quad p\geq 1.
\end{equation}
\end{Prop}
\begin{proof}
It follows from Lemma \ref{AL} that it is enough to prove this lemma for $m = 0$. For $\Phi_{\qa,0}^0$, it is easy to see from \eqref{AK} that there exist differential polynomials $g_{\qa}\in\hm A^0_{
\geq 1}$ such that
\begin{equation}
\label{AM}
h_{\qa,0} = \eta_{\qa\qb}v^\qb+\qp_xg_{\qa}.
\end{equation}
Therefore we arrive at
\begin{equation}
\label{BL}
\Phi_{\qa,0}^0 = \qs_{\qa,0}+\diff{g_\qa}{\qt_0}\in\hm A.
\end{equation}
We proceed to consider $\Phi_{\qa,p}^0$ for $p\geq 1$. By using Proposition \ref{AF} we can rewrite \eqref{AN} as follows:
\begin{equation}
\label{AO}
\fk{\int h_{\qa,p}}{P_1}=\left(p+\frac12+\mu_\qa\right)\fk{\int h_{\qa,p+1}}{P_0}+\sum_{k=1}^p \left(R_k\right)^\qg_\qa\fk{\int h_{\qg,p-k+1}}{P_0}.
\end{equation}
Take $p=0$ in\eqref{AO} and rewrite it in the form
\[
\kk{\frac 12+\mu_\qa}\int\diff{h_{\qa,1}}{\qt_0} = \int\diff{h_{\qa,0}}{\qt_1},
\]
which means that there exists a differential polynomial $p_{\qa,1}\in\hm A^1$ such that
\[
\kk{\frac 12+\mu_\qa}\diff{h_{\qa,1}}{\qt_0} = \diff{h_{\qa,0}}{\qt_1}+\qp_x p_{\qa,1}.
\]
Therefore we have
\[
\kk{\frac 12+\mu_\qa}\Phi_{\qa,1}^0 = \Phi_{\qa,0}^1+p_{\qa,1}\in\hm A^+.
\]
For general $p\geq 1$, we can prove \eqref{DM} by using \eqref{AO} and induction on $p$. The proposition is proved.
\end{proof}

\section{Deformation of Virasoro symmetries: formulation}
\label{sec3}
In this section, we first recall the theory of variational bihamiltonian cohomology given in \cite{liu2021variational} and then explain how to use it to study Virasoro symmetries of the deformed Principal Hierarchies. We also use the example of the  deformation of the Riemann hierarchy to illustrate our approach to the study of Virasoro symmetries.

\subsection{Variational bihamiltonian cohomologies}
In \cite{liu2021variational}, we established a cohomology theory on the space $\derx$ consisting of derivations on $\hm A$ that commute with $\qp_x$. This theory provides us suitable tools to study Virasoro symmetries of deformations of the Principal Hierarchies. We recall the basic definitions and results in this subsection.

Let us define the space  $\derx$ by
\[
\derx = \{X\in\der\mid [X,\qp_x] = 0\}.
\]
We grade the space $\derx$ by the differential gradation and the super gradation, and use the notation ${\derx}^p_d$ to denote the subspace of homogeneous elements with super degree $p$ and differential degree $d$.
\begin{Lem}
${\derx}^p_d = 0$ for $p\leq -2$ or $d<0$.
\end{Lem}
\begin{proof}
Let us choose a local coordinate system $(w^\qa,\phi_\qa)$ on $\hat M$. Assume $X\in\derx$ with super degree $p\leq -2$ or $d< 0$. Then by definition this means 
\[
X(w^\qa) = X(\phi_\qa) = 0.
\]
Since $[X,\qp_x] = 0$, we immediately see that $X({w^{\qa,s}}) = 0$ and $X(\phi_\qa^s) = 0$ for $s\geq 0$. Hence $X = 0$ and the lemma is proved.
\end{proof}
Let $P^{[0]}$ be a Hamiltonian structure of hydrodynamic type and $(P_0^{[0]},P_1^{[0]})$ be a semisimple bihamiltonian structure of hydrodynamic type. Then by using \eqref{dc} we have a complex $(\derx,D_{P^{[0]}})$ and a double complex $(\derx,D_{P_0^{[0]}},D_{P_1^{[0]}})$. We define the following cohomology groups:
\begin{equation}
\label{AQ}
H^p_d\bigl(\derx,P^{[0]}\bigr) = \frac{{\derx}^p_d\cap \ker D_{P^{[0]}}}{{\derx}^p_d\cap\ima D_{P^{[0]}}},\quad p,\,d\geq 0,
\end{equation}
\begin{equation}
\label{AP}
BH^p_d\bigl(\derx,P_0^{[0]},P_1^{[0]}\bigr) = \frac{{\derx}^p_d\cap \ker D_{P_0^{[0]}}\cap\ker D_{P_1^{[0]}}}{{\derx}^p_d\cap\ima D_{P_0^{[0]}}D_{P_1^{[0]}}}, \quad p,\,d\geq 0.
\end{equation}
Note that the spaces ${\derx}^{-1}_d \neq 0$ for $d\geq 0$ and they must be taken into account while computing the cohomology groups. For example, the space $H^0_d(\derx,P^{[0]})$ is given by
\[
H^0_d\bigl(\derx,P^{[0]}\bigr) = \frac{\ker(D_{P^{[0]}}:{\derx}^{0}_d\to {\derx}^{1}_{d+1})}{\ima(D_{P^{[0]}}:{\derx}^{-1}_{d-1}\to {\derx}^{0}_{d})}.
\]

By using the canonical symplectic structure on $\hat M$, we can identify the space $\derx$ with the space $\bar\Qo$ of local functionals of variational 1-forms, and the vector fields $D_{P^{[0]}}$,  $D_{P_0^{[0]}}$ and $D_{P_1^{[0]}}$ induce Lie derivatives on $\bar\Qo$. This is the reason why we call the above cohomology groups the variational cohomology groups. In \cite{liu2021variational}, the cohomology groups \eqref{AQ} and \eqref{AP} are computed by converting them to the cohomology groups on the space $\bar\Qo$. The details of the computation of these cohomology groups are not used in the present paper,  so we omit them and refer the readers to \cite{liu2021variational}. The following result plays an essential role in the present paper.
\begin{Th}[\cite{liu2021variational}]
\label{AR}
We have the following results on the cohomology groups \eqref{AQ} and \eqref{AP}:
\begin{enumerate}
\item $H^p_d\bigl(\derx,P^{[0]}\bigr) = 0$ for $p\geq 0$, $d> 0$.
\item $BH^0_{\geq 2}\bigl(\derx,P_0^{[0]},P_1^{[0]}\bigr) = 0$. \item $BH^1_{\geq 4}\bigl(\derx,P_0^{[0]},P_1^{[0]}\bigr) = 0$.
\item $BH^1_3\bigl(\derx,P_0^{[0]},P_1^{[0]}\bigr) \cong \oplus_{i=1}^nC^\infty(\mathbb R)$, and if we denote the action  of a cocycle $X\in {\derx}^1_3$ on the $i$-th canonical coordinate $u^i$ by
\[
X(u^i) = \sum_{j=1}^n\sum_{k=0}^3 X_{i,j}^k\qth_j^{3-k},\quad  X_{i,j}^k\in\hm A^0_k,
\]
then the cohomology class $[X]$ is determined by the following  functions:
\[
c_1 = \frac{X_{1,1}^0}{(f^1)^2},\quad c_2 = \frac{X_{2,2}^0}{(f^2)^2},\quad\cdots,\quad c_n = \frac{X_{n,n}^0}{(f^n)^2}.
\]
Here each function $c_i$  depends only on the $i$-th canonical coordinate $u^i$, and $f^i$ is the function defined in \eqref{CW}.
\end{enumerate}
\end{Th}

\subsection{Virasoro symmetries of the Principal Hierarchies}
In this subsection, we recall the construction of  Virasoro symmetries of the super tau-cover of the Principal Hierarchy following \cite{liu2020super}. In \cite{dubrovin1999frobenius}, a family of infinitely many symmetries $\diff{}{s_m^{even}}$ for $m\geq -1$ of the tau-cover of the Principal Hierarchy associated with a Frobenius manifold $M$ was constructed. This family of symmetries are called the Virasoro symmetries due to the property
\[
\fk{\diff{}{s_k^{even}}}{\diff{}{s_l^{even}}} = (l-k)\diff{}{s_{k+l}^{even}},\quad k,l\geq -1.
\]
These symmetries can be represented by a family of quadratic differential operators $L_m^{even}$ of the form 
\[L^{even}_m=a_m^{\qa,p;\qb,q}\frac{\qp^2}{\qp t^{\qa,p}\qp t^{\qb,q}}+{b}_{m;\qa,p}^{\qb,q} t^{\qa,p}\frac{\qp}{\qp t^{\qb,q}}+c_{m;\qa,p;\qb,q} t^{\qa,p} t^{\qb,q}+\frac{1}{4}\qd_{m,0}\mathrm{tr}\left(\frac{1}{4}-\mu^2\right),
\]
where $a_m^{\qa,p;\qb,q}$, ${b}_{m;\qa,p}^{\qb,q}$, $c_{m;\qa,p;\qb,q}$ are some constants determined from the monodromy data of $M$ and one may refer to \cite{dubrovin1999frobenius} for details. These operators satisfy the Virasoro commutation relation
\[
[L_k^{even},L_l^{even}] = (k-l)L_{l+k}^{even}.
\]
In this paper, we only need the explicit expressions for $L_{-1}^{even}$ and $L_2^{even}$ which are given by
\begin{align}\label{L-1}
L_{-1}^{even} &= \frac 12 \eta_{\qa\qb}t^{\qa,0}t^{\qb,0}+\sum_{p\geq 1}t^{\qa,p}\diff{}{t^{\qa,p-1}},\\
\label{L2}
L_2^{even} &= a^{\qa\qb}\frac{\qp^2}{\qp t^{\qa,1}\qp t^{\qb,0}}+b^{\qa\qb}\frac{\qp^2}{\qp t^{\qa,0}\qp t^{\qb,0}}+\mathcal L_2^{even} + c_{2;\qa,p;\qb,q}t^{\qa,p}t^{\qb,q},
\end{align}
where the constants have the expressions
\begin{equation}
\label{BT}
a^{\qa\qb} = \eta^{\qa\qb}\kk{\frac 12+\mu_\qb}\kk{\frac 12+\mu_\qa}\kk{\frac 32+\mu_\qa},
\end{equation}
\begin{equation}
\label{BU}
b^{\qa\qb} = \frac 12 \eta^{\qb\qg}(R_1)^\qa_\qg\kk{\frac 14+3\mu_\qb-3\mu_{\qb}^2},
\end{equation}
and the operator $\mathcal L_2^{even}$ is given by
\begin{align}
\label{AS}
& \sum_{p\geq 0}\kk{p+\frac 12+\mu_\qa}\kk{p+\frac 32+\mu_\qa}\kk{p+\frac 52+\mu_\qa}t^{\qa,p}\diff{}{t^{\qa,p+2}}\\
\notag
&+\sum_{p\geq 0}\sum_{1\leq r\leq p+2}\kk{3\kk{p+\frac 12+\mu_\qa}^2+6\kk{p+\frac 12+\mu_\qa}+2}\kk{R_r}^\qb_\qa t^{\qa,p}\diff{}{t^{\qb,p-r+2}}\\
\notag
&+\sum_{p\geq 0}\sum_{2\leq r\leq p+2}\kk{3p+\frac 92+3\mu_\qa}\kk{R_{r,2}}^\qb_\qa t^{\qa,p}\diff{}{t^{\qb,p-r+2}}\\
\notag
&+\sum_{p\geq 1}\sum_{3\leq r\leq p+2}\kk{R_{r,3}}^\qb_\qa t^{\qa,p}\diff{}{t^{\qb,p-r+2}}.
\end{align}
The explicit expressions for the matrices $R_{k,l}$ and constants $c_{2;\qa,p;\qb,q}$ are not used in this paper, so we omit them.

We have the following theorem for the Virasoro symmetries of the tau-cover of the Principal Hierarchy.
\begin{Th}[\cite{dubrovin1999frobenius}]
\label{AT}
Let us define the following time-dependent flows for $m\geq -1$:
\begin{align*}
\diff{f_{\ql,k}}{s_m^{even}} =& \diff{}{t^{\ql,k}}\kk{\sum a_m^{\qa,p;\qb,q}f_{\qa,p}f_{\qb,q}+\sum b_{m;\qa,p}^{\qb,q} t^{\qa,p}f_{\qb,q}+\sum c_{m;\qa,p;\qb,q} t^{\qa,p} t^{\qb,q}},\\
\diff{v^{\ql}}{s_m^{even}} =& \eta^{\ql\qg}\diff{}{t^{1,0}}\diff{f_{\qg,0}}{s_m^{even}}.
\end{align*}
Then the following commutation relations hold true:
\[
\fk{\diff{}{s_k^{even}}}{\diff{}{t^{\qa,p}}} = 0,\quad \fk{\diff{}{s_k^{even}}}{\diff{}{s_m^{even}}} = (m-k)\diff{}{s_{k+l}^{even}},\quad k,m\geq -1.
\]
\end{Th}

We also have the following theorem for the Virasoro symmetries of the super tau-cover of the Principal Hierarchy.

\begin{Th}[\cite{liu2020super}]
\label{BA}
Let us denote
\[
\diff{}{s_m^{odd}} = \sum_{p\geq 0}(p+c_0)\qt_p\diff{}{\qt_{p+m}},\quad m \geq -1,
\]
where $c_0$ is an arbitrary constant. Then the following flows are symmetries of the super tau-cover of the Principal Hierarchy  associated with a Frobenius manifold:
\begin{align*}
\frac{\qp f_{\qa, p}}{\qp s_m}&=
\frac{\qp f_{\qa, p}}{\qp s^{even}_m}+\frac{\qp f_{\qa, p}}{\qp s^{odd}_m},
\quad \diff{\Phi^n_{\qa,p}}{s_m}=\frac{\qp}{\qp \tau_n}\left(\frac{\qp f_{\qa,p}}{\qp s_m}\right),\\
\frac{\qp v^\qa}{\qp s_m}&=
\frac{\qp v^\qa}{\qp s^{even}_m}+\frac{\qp v^\qa}{\qp s^{odd}_m},\quad \frac{\qp \sigma_{\qa,p}}{\qp s_m}=\frac{\qp}{\qp \tau_p}\left(\frac{\qp f_{\qa,0}}{\qp s_m}\right).
\end{align*}
Moreover, these flows satisfy the commutation relation
\[\left[\frac{\qp}{\qp s_k},\frac{\qp}{\qp s_m}\right]=(m-k) \frac{\qp}{\qp s_{k+m}},\quad k, m\ge -1.\]
\end{Th}

\subsection{Formulation of the deformation problem}
\label{AU}
In this subsection, we  first state the main problem of this paper, then explain the motivation and strategy of our proof. 

From now on, we fix a semisimple Frobenius manifold $M$ of dimension $n$ and let $(P_0,P_1)$ be a deformation of the bihamiltonian structure \eqref{biham-frob} with constant central invariants. Then $(P_0,P_1)$ determines a unique deformation of the Principal Hierarchy associated with $M$. After a suitable Miura type transformation we may assume
\[
P_0  = \frac 12\int\eta^{\qa\qb}\qs_\qa\qs_\qb^1 = [Z,P_1],\quad Z = \int \qs_1,
\]
and $P_1$ has no $\hm F^2_2$ components. We also introduce odd variables $\qs_{\qa,m}$ for $m\geq 0$ as we explained in Sect.\,\ref{AV}.

Our problem can be stated as follows: to find a unique derivation $X\in\derx^0_{\geq 0}$ such that the flow defined by
\begin{align}
\label{BH}
\diff{v_\ql}{s_2} =\ & a^{\qa\qb}\kk{f_{\qb,0}\diff{v_\ql}{t^{\qa,1}}+f_{\qa,1}\diff{v_\ql}{t^{\qb,0}}}+ b^{\qa\qb}\kk{f_{\qb,0}\diff{v_\ql}{t^{\qa,0}}+f_{\qa,0}\diff{v_\ql}{t^{\qb,0}}}\\\notag&+Xv_{\ql}+\mathcal L_2v_\ql,\\
\label{BI}
\diff{\qs_{\ql,0}}{s_2} =\ & a^{\qa\qb}\kk{f_{\qb,0}\diff{\qs_{\ql,0}}{t^{\qa,1}}+f_{\qa,1}\diff{\qs_{\ql,0}}{t^{\qb,0}}}+ b^{\qa\qb}\kk{f_{\qb,0}\diff{\qs_{\ql,0}}{t^{\qa,0}}+f_{\qa,0}\diff{\qs_{\ql,0}}{t^{\qb,0}}}\\\notag&+X\qs_{\ql,0}
+M^\qz_{\ql}\qs_{\qz,2}+N^\qz_{\ql}\qs_{\qz,1}+\mathcal L_2\qs_{\ql,0},
\end{align}
satisfies the conditions
\begin{equation}
\label{AX}
\fk{\diff{}{s_2}}{\diff{}{\qt_0}} = \fk{\diff{}{s_2}}{\diff{}{\qt_1}} = 0,
\end{equation}
and we require that the leading term of $X$ is determined by Virasoro symmetry of the super tau-cover of the Principal Hierarchy.
Here the operator $\mathcal L_2$ is given by
\[
\mathcal L_2 = \mathcal L_2^{even}+\sum_{p\geq 0}(p+c_0)\qt_p\diff{}{\qt_{p+2}},
\]
and $M^\qz_\ql,\ N^\qz_\ql\in\hm A^0$ are some differential polynomials whose definitions will be given later.

If we find such a derivation $X$, then we can prove that the flow $\diff{}{s_2}$ is a symmetry of the deformed Principal Hierarchy, i.e., 
\[
\fk{\diff{}{s_2}}{\diff{}{t^{\qa,p}}} = 0,\quad \qa = 1,\cdots, n,\quad p\geq 0.
\]
To prove this fact, we only need to show that
\begin{equation}\label{AY}
\fk{\diff{}{s_2}}{\diff{}{t^{\qa,p}}}\in\der^0_{\geq 2},
\end{equation}
since this commutator is a cocycle in $(\derx^0_{\geq 2},D_{P_0},D_{P_1})$, and by using the property $BH^0_{\geq 2}\bigl(\derx,D_{P_0^{[0]}},D_{P_1^{[0]}}\bigr) = 0$ it is easy to conclude that this commutator vanishes. Note that a priori we have
\[
\fk{\diff{}{s_2}}{\diff{}{t^{\qa,p}}}v_\ql\in\hm A,\quad \fk{\diff{}{s_2}}{\diff{}{t^{\qa,p}}}\qs_{\ql,0}\in\hm A^+,
\]
so the condition \eqref{AY} is actually a \textit{locality condition}, i.e., it is equivalent to
\begin{equation}
\label{AZ}
\fk{\diff{}{s_2}}{\diff{}{t^{\qa,p}}}v_\ql\in\hm A,\quad \fk{\diff{}{s_2}}{\diff{}{t^{\qa,p}}}\qs_{\ql,0}\in\hm A,
\end{equation}
which is the condition we actually need to check.

Let us explain how to find a unique $X\in \derx$ such that the conditions \eqref{AX} hold true. To this end, we first rewrite the conditions \eqref{AX} into the equations for $X$ as follows:
\begin{equation}
\label{BD}
\fk{\diff{}{\qt_0}}{X} = I_0,\quad \fk{\diff{}{\qt_1}}{X} = I_1,
\end{equation}
where $I_0$ and $I_1$ are some derivations that will be given later. The uniqueness of $X$ is a consequence of the equations \eqref{BD} and the fact that $BH^0_{\geq 2}\bigl(\derx\bigr) = 0$, since the leading term of $X$ is fixed. We will prove the existence of $X$ by using the following steps.

\textbf{Step 1.} To check the \textit{locality condition} \eqref{AZ} and to prove that
\begin{equation}
\label{BB}
I_0,\  I_1\in\der.
\end{equation}

\textbf{Step 2.} To check the \textit{closedness condition}
\begin{equation}
\label{BC}
\fk{\diff{}{\qt_0}}{I_0} = 0.
\end{equation}

When Step 2 is finished, we can find a derivation $X^\circ\in\derx$ by using the property $H^1_{>0}\bigl(\derx\bigr) = 0$ such that
\[
I_0 = \fk{\diff{}{\qt_0}}{X^\circ},
\]
and the leading term of $X^\circ$ is given by the Virasoro symmetry of the super tau-cover of the Principal Hierarchy.
We define $\mathcal C = X-X^\circ\in\derx^0_{\geq 2}$, then the equations \eqref{BD} for $X$ are transformed to the following equations for $\mathcal C$:
\begin{equation}
\label{BE}
\fk{\diff{}{\qt_0}}{\mathcal C} = 0,\quad \fk{\diff{}{\qt_1}}{\mathcal C} = I_1-\fk{\diff{}{\qt_1}}{X^\circ}.
\end{equation}

\textbf{Step 3.} To check the closedness conditions
\begin{equation}
\label{BF}
\fk{\diff{}{\qt_0}}{I_1-\fk{\diff{}{\qt_1}}{X^\circ}} = 0,\quad \fk{\diff{}{\qt_1}}{I_1} = 0. 
\end{equation}

\textbf{Step 4.} To check that the differential degree $3$ component of the derivation $I_1-\fk{\diff{}{\qt_1}}{X^\circ}$ vanishes in the cohomology group $BH^1_{3}\bigl(\derx\bigr)$.

We call the above fact the \textit{vanishing of the genus one obstruction} for the following reason. By using the first equation in \eqref{BE} and the vanishing of $H^0_{\geq 2}\bigl(\derx\bigr)$, we see that there exists a unique $\mathcal T\in\derx^{-1}_{\geq 1}$ such that
\[
\mathcal C = \fk{\diff{}{\qt_0}}{\mathcal T}.
\]
The derivation $\mathcal T$ must also satisfy the second equation in \eqref{BE}
\begin{equation}
\label{BG}
 \fk{\diff{}{\qt_1}}{ \fk{\diff{}{\qt_0}}{\mathcal T}} = I_1-\fk{\diff{}{\qt_1}}{X^\circ}.
\end{equation}
If the differential degree $3$ component of the derivation $I_1-\fk{\diff{}{\qt_1}}{X^\circ}$ does not vanish in the cohomology group $BH^1_{3}\bigl(\derx\bigr)$, then such $\mathcal T$ does not exist. 

However if the genus one obstruction vanishes, the derivation $\mathcal T$ exists upto differential degree $1$. Then by using $BH^1_{\geq 4}\bigl(\derx\bigr) = 0$ and the closedness conditions \eqref{BF}, we can solve $\mathcal T$ from \eqref{BG} degree by degree and hence we can solve $X$ such that \eqref{AX} holds true.

\textbf{Step 5.} To lift the symmetry $\diff{v_\ql}{s_2}$ to a symmetry of the tau-cover of the deformed Principal Hierarchy, and to define all the other flows $\diff{}{s_m}$ for $m\geq 0$ for the Virasoro symmetries. Note that the symmetry $\diff{}{s_{-1}}$ is constructed in \cite{dubrovin2018bihamiltonian}. We remark that we can also lift the symmetry \eqref{BH} and \eqref{BI} to a symmetry of the super tau-cover of the deformed Principal Hierarchy, but it is not necessary for the consideration for our problem.

\vskip 1em
In the remaining part of this subsection, we explain how the equations \eqref{BH} and \eqref{BI} are derived from the equation \eqref{DN} of the main theorem. Let $\mathcal Z$ be a tau-function of the tau-cover \eqref{AW} of the deformed Principal Hierarchy, i.e., 
\begin{equation}\label{BK}
f_{\qa,p} = \diff{\log\mathcal  Z}{t^{\qa,p}},\quad w^\qa = \eta^{\qa\qb}\frac{\qp^2\log\mathcal Z}{\qp t^{1,0}\qp t^{\qb,0}}
\end{equation}
give a solution of the tau-cover \eqref{AW}. Our goal is to find  a symmetry of the following form
\begin{equation}
\label{BN}
\diff{\mathcal Z}{s_2} = L_2^{even}\mathcal Z+ O_2\mathcal Z,
\end{equation}
where $L_2^{even}$ is the operator \eqref{L2} and $O_2$ is a differential polynomial. If we assume that this is indeed a symmetry, then by using \eqref{BK} we obtain the flow
\begin{align*}
\diff{f_{\ql,k}}{s_2} =& a^{\qa\qb}\kk{f_{\qa,1}\Qo_{\qb,0;\ql,k}+f_{\qb,0}\Qo_{\qa,1;\ql,k}+\diff{\Qo_{\qa,1;\qb,0}}{t^{\ql,k}}}\\
&+b^{\qa\qb}\kk{f_{\qa,0}\Qo_{\qb,0;\ql,k}+f_{\qb,0}\Qo_{\qa,0;\ql,k}+\diff{\Qo_{\qa,0;\qb,0}}{t^{\ql,k}}}\\
&+b^{\qb,q}_{2;\ql,k}f_{\qb,q}+c_{2;\ql,k;\qb,q}t^{\qb,q}+c_{2;\qb,q;\ql,k}t^{\qb,q}+\diff{O_2}{t^{\ql,k}}+\mathcal L_2^{even} f_{\ql,k};\\
\diff{w_{\ql}}{s_2} =& a^{\qa\qb}\kk{f_{\qb,0}\diff{w_\ql}{t^{\qa,1}}+f_{\qa,1}\diff{w_\ql}{t^{\qb,0}}}+ b^{\qa\qb}\kk{f_{\qb,0}\diff{w_\ql}{t^{\qa,0}}+f_{\qa,0}\diff{w_\ql}{t^{\qb,0}}}\\+&W_{\ql}+\mathcal L_2^{even}w_\ql,
\end{align*}
here $W_\ql$ are some differential polynomials. Now recall that $v_\ql$ and $w_\qz$ are related by a Miura type transformation, hence by using the equation
\[
\diff{v_\ql}{s_2} = \sum_{s\geq 0}\qp_x^s\diff{w_{\qz}}{s_2}\diff{v_\ql}{w_\qz^{(s)}}
\]
we know that there exist differential polynomials $X^0_\ql\in\hm A^0$ such that
\begin{align}
\label{DO}
\diff{v_\ql}{s_2} =&\,a^{\qa\qb}\kk{f_{\qb,0}\diff{v_\ql}{t^{\qa,1}}+f_{\qa,1}\diff{v_\ql}{t^{\qb,0}}}+ b^{\qa\qb}\kk{f_{\qb,0}\diff{v_\ql}{t^{\qa,0}}+f_{\qa,0}\diff{v_\ql}{t^{\qb,0}}}\\
\notag
&+X_{\ql}^0+\mathcal L_2^{even}v_\ql.
\end{align}
To write down the evolutions of the odd variables, we must replace $\mathcal L_2^{even}$ by $\mathcal L_2$. By using \eqref{BL} it is easy to see that 
\[
\diff{\boldsymbol{f_0}}{\qt_0} = \boldsymbol A\boldsymbol{\qs_0},
\]
where $\boldsymbol{f_0} = (f_{1,0},\cdots,f_{n,0})^T$, $\boldsymbol{\qs_0} = (\qs_{1,0},\cdots,\qs_{n,0})^T$ and $\boldsymbol A$ is a matrix of differential operator of the form
\[
\boldsymbol A = \sum_{g\geq 0}\sum_{k=0}^{2g}\boldsymbol {A_{g,k}}\qp_x^k,\quad \boldsymbol {A_{g,k}} \in M_n(\hm A^0_{2g-k}).
\]
Note that $\boldsymbol {A_{0,0}}$ is the identity matrix,  therefore $\boldsymbol A$ is invertible as a differential operator, i.e., there exists $\boldsymbol B = \sum_{g\geq 0}\sum_{k=0}^{2g}\boldsymbol {B_{g,k}}\qp_x^k$ such that $\boldsymbol{AB} = \boldsymbol{BA} = \boldsymbol{I}$, and in particular $\boldsymbol {B_{0,0}}$ is the identity matrix. Thus we can represent the odd variable $\qs_{\ql,0}$ in the form
\[
\qs_{\ql,0} = \diff{f_{\ql,0}}{\qt_0}+\sum_{g\geq 1}\sum_{k=0}^{2g} {\boldsymbol {B_{g,k}}}_\ql^\qz\qp_x^k\diff{f_{\qz,0}}{\qt_0}.
\]
This identity leads us to the following definition of the evolutions of the odd variables $\qs_{\ql,0}$ along the flow $\diff{}{s_2}$:
\[
\diff{\qs_{\ql,0}}{s_2} = \diff{}{\qt_0}\diff{f_{\ql,0}}{s_2}+\sum_{g\geq 1}\sum_{k=0}^{2g} \diff{}{s_2}\kk{{\boldsymbol {B_{g,k}}}_\ql^\qz}\qp_x^k\diff{f_{\qz,0}}{\qt_0}+\sum_{g\geq 1}\sum_{k=0}^{2g} {\boldsymbol {B_{g,k}}}_\ql^\qz\qp_x^k \diff{}{\qt_0}\diff{f_{\qz,0}}{s_2}.
\]

By using Proposition \ref{BM} and Lemma \ref{AL}, we can see that the odd variables $\Phi_{\qa,p}^m$ appearing in $ \diff{}{\qt_0}\diff{f_{\ql,0}}{s_2}$ can be represented by elements of $\hm A^+$. So there exist differential polynomials $M^\qz_\ql$, $N^\qz_\ql\in\hm A^0$ and $X_\ql^1\in\hm A^1$ such that
\begin{align*}
\diff{\qs_{\ql,0}}{s_2} = &\, a^{\qa\qb}\kk{f_{\qb,0}\diff{\qs_{\ql,0}}{t^{\qa,1}}+f_{\qa,1}\diff{\qs_{\ql,0}}{t^{\qb,0}}}+ b^{\qa\qb}\kk{f_{\qb,0}\diff{\qs_{\ql,0}}{t^{\qa,0}}+f_{\qa,0}\diff{\qs_{\ql,0}}{t^{\qb,0}}}\\&+X^1_{\ql}
+M^\qz_{\ql}\qs_{\qz,2}+N^\qz_{\ql}\qs_{\qz,1}+\mathcal L_2\qs_{\ql,0}.
\end{align*}
Finally we define a derivation $X\in\derx^0$ such that $Xv_{\ql} = X_\ql^0$, where $X_\ql^0$ is the differential polynomial introduced in \eqref{DO}, and $X\qs_{\ql,0} = X^1_\ql$. Therefore our problem of finding such a derivation $X$ is a necessary condition of the main theorem.

\subsection{Example: one-dimensional Frobenius manifold}
\label{B0}
In this subsection we present an example to illustrate how the general framework described in the previous subsection works. We consider the one-dimensional Frobenius manifold $M$, it has the following potential and Euler vector field:
\[
F = \frac 16 v^3,\quad E = v\qp_v.
\]
Due to the dimension reason, we will omit the Greek indices, for example, we will use $v^{(s)}$ and $\qs_m^s$ instead of $v^{1,s}$ and $\qs_{1,m}^s$. The Principal Hierarchy associated with $M$ is the Riemann hierarchy
\[
\diff{v}{t_p} = \frac{v^p}{p!}v_x,\quad p\geq 0,
\] 
whose bihamiltonian structure is given by
\[
P_0^{[0]} = \frac 12\int\qs_0\qs_0^1,\quad P_1^{[0]} = \frac 12\int v\qs_0\qs_0^1.
\]

It is proved in \cite{DLZ-1,liu2005deformations} that every deformation $(P_0,P_1)$ of $(P_0^{[0]},P_1^{[0]})$ with a constant central invariant is equivalent to the bihamiltonian structure given by
\[
P_0 = \frac 12\int\qs_0\qs_0^1,\quad P_1 = \frac 12\int v\qs_0\qs_0^1 + \qe^2c\qs_0\qs_0^3
\]
via a certain  Miura type transformation. Here the dispersion parameter $\qe$ is added for clearness, and the central invariant of $(P_0,P_1)$ is $\frac{c}{3}$. In particular, when $c = \frac 18$ the corresponding deformed Riemann Hierarchy is the KdV hierarchy that controls the 2D  gravity \cite{kontsevich1992intersection,witten1990two}.

We introduce the odd variables $\qs_m$ satisfying the recursion relations
\[
\qs_{m+1}^1 = v\qs_{m}^1+\frac 12v_x\qs_m+\qe^2 c \qs_m^3,\quad m\geq 0.
\]
The following useful results can be obtained directly:
\begin{align*}
h_1 &= \frac{v^2}{2}+\frac 23 \qe^2cv_{xx},\quad h_2 = \frac{v^3}{6}+\qe^2c\kk{\frac 13 v_x^2+\frac 23vv_{xx}}+\frac{4}{15}\qe^4c^2v^{(4)};\\
\diff{v}{t_1} &= vv_x+\frac 23\qe^2cv^{(3)},\quad \diff{\qs_0}{t_1} = v\qs_0^1+\frac 23\qe^2c\qs_0^3;\\
\diff{v}{t_2}&= \frac 12 v^2v_x+\qe^2c\kk{\frac 43 v_xv_{xx}+\frac 23 vv^{(3)}}+\frac{4}{15}\qe^4c^2v^{(5)},\\
\diff{\qs_0}{t_2}& = \frac 12 v^2\qs_0^1+\qe^2c\kk{\frac 23 v_{xx}\qs_0^1+\frac23 v_x\qs_0^2+\frac 23 v\qs_0^3}+\frac{4}{15}\qe^4c^2\qs_0^5;\\
\qe\diff{f_0}{\qt_0} &= \qs_0,\quad \qe\diff{f_1}{\qt_0} = 2\qs_1-v\qs_0-\frac 43\qe^2c\qs_0^2;\\
\qe\diff{f_2}{\qt_0} &= \frac 43\qs_2-\frac 23v\qs_1\frac 16 v^2\qs_0-\qe^2c\kk{\frac 23v_{xx}\qs_0+\frac 43v_x\qs_0^2+\frac 43v\qs_0^2}-\frac{16}{15}\qe^4c^2\qs_0^4.
\end{align*}

We first note that
\[
L_2 = \frac 38\qe^2\frac{\qp^2}{\qp t_1\qp t_0}+\mathcal L_2,\quad \mathcal L_2 = \sum_{p\geq 0}\frac{\Qg\kk{\frac 72+p}}{\Qg\kk{\frac 12+p}}t_p\diff{}{t_{p+2}}+(p+c_0)\qt_p\diff{}{\qt_{p+2}}.
\]
Then the equations \eqref{BH} and \eqref{BI} for this example have the following form: 
\begin{align*}
\diff{v}{s_2} =&\ \frac 38\qe\kk{v_xf_1+\diff{v}{t_1}f_0}+Xv+\mathcal L_2v,\\
\diff{\qs_0}{s_2} =&\ \frac 38\qe\kk{\qs_0^1f_1+\diff{\qs_0}{t_1}f_0}+\kk{\frac 52+c_0}\qs_2-\frac v2\qs_1+X\qs_0+\mathcal L_2\qs_0.
\end{align*}
We are to find the derivation $X\in\derx^0$ such that the flow $\diff{}{s_2}$ commutes with $\diff{}{\qt_0}$ and $\diff{}{\qt_1}$. These conditions yield the following equations for $X$:
\[
\fk{\diff{}{\qt_0}}{X} = I_0,\quad \fk{\diff{}{\qt_1}}{X} = I_1,
\]
where the derivations $I_0$ and $I_1$ are defined by
\begin{align}
\label{BP}
I_0v=&\ vv_x\qs_0+\frac 72v^2\qs_0^1+\qe^2c\kk{v^{(3)}\qs_0+\frac{13}{2}v_{xx}\qs_0^1+8v_x\qs_0^2+6v\qs_0^3}+3\qe^4c^2\qs_0^5,\\
\label{BQ}
I_0\qs_0=&\ v\qs_0\qs_0^1-\qe^2c\kk{\frac 12 \qs_0^1\qs_0^2-\qs_0\qs_0^3},
\end{align}
\begin{align}
\label{BR}
I_1v =&\ \frac 54v^2v_x\qs_0+\frac 52v^3\qs_0^1+\qe^2c\kk{\frac 72 v_xv_{xx}\qs_0+2vv^{(3)}\qs_0+\frac{45}{4}v_x^2\qs_0^1+\frac{31}{2}vv_{xx}\qs_0^1}\\
\notag &+\qe^2c\kk{26vv_x\qs_0^2+\frac{19}{2}v^2\qs_0^3}+\qe^4c^2\kk{v^{(5)}\qs_0+\frac{17}{2}v^{(4)}\qs_0^1+\frac{45}{2}v^{(3)}\qs_0^2}\\
\notag &+\qe^4c^2\kk{\frac{59}{2}v_{xx}\qs_0^3+\frac{43}{2}v_x\qs_0^4+9v\qs_0^5}+3\qe^6c^3\qs_0^7,\\
\label{BS}
I_1\qs_0 =&\ \frac 54v^2\qs_0\qs_0^1+\qe^2c\kk{\frac 52 v_{xx}\qs_0\qs_0^1+\frac 52 v_x\qs_0\qs_0^2-\frac 12v\qs_0^1\qs_0^2+2v\qs_0\qs_0^3}\\
\notag &-\qe^4c^2\kk{\frac 12\qs_0^1\qs_0^4-\qs_0\qs_0^5}.
\end{align}

It follows from \eqref{BP} and \eqref{BQ} that we can choose  $X^\circ\in\derx$ such that $[\diff{}{\qt_0},X^\circ] = I_0$,  whose actions on $v$ and $\qs_0$ are given by
\begin{align*}
X^\circ v &= v^3+\qe^2c\kk{\frac 54 v_x^2+3vv_{xx}}+\qe^4c^2v^{(4)},\\
X^\circ\qs_0 &= -\frac 12 v^2\qs_0-\qe^2c\kk{v_{xx}\qs_0+\frac 52 v_x\qs_0^1+3v\qs_0^2}-2\qe^4c^2\qs_0^4. 
\end{align*}
Then according to the general discussion given in the previous subsection, the derivation $\mathcal C = X-X^\circ$ satisfies the equations
\[
\fk{\diff{}{\qt_0}}{\mathcal C} = 0,\quad \fk{\diff{}{\qt_1}}{\mathcal C} = I_1-\fk{\diff{}{\qt_1}}{X^\circ}.
\]
Finally we can solve the above equations and obtain the unique derivation $\mathcal C$ that is defined by
\begin{align*}
\mathcal Cv &= \qe^2c\kk{3v_x^2+3vv_{xx}}+2\qe^4c^2v^{(4)},\\
\mathcal C\qs_0 &= \qe^2c\kk{3v_x\qs_0^1+v\qs_0^2}+2\qe^4c^2\qs_0^4.
\end{align*}
By forgetting all the odd variables we obtain the following  symmetry for the deformed Riemann hierarchy:
\[
\diff{v}{s_2} = \frac 38\qe\kk{v_xf_1+\diff{v}{t_1}f_0}+v^3+\qe^2c\kk{\frac{17}{4}v_x^2+6vv_{xx}}+3\qe^4c^2v^{(4)}+\mathcal L_2^{even} v.
\]
It is easy to check that the action of this symmetry on
the tau function $\mathcal Z$ of the deformed Riemann hierarchy  can be represented by
\begin{equation}
\label{DE}
\diff{\mathcal Z}{s_2} = L_2^{even}\mathcal Z+\kk{3c-\frac 38}\kk{\frac{v^2}{2}+\frac 23\qe^2cv_{xx}}\mathcal Z.
\end{equation}
In particular, when $c = \frac 18$, this symmetry is given by a linear action on $\mathcal Z$.

\section{Deformation of Virasoro symmetries: existence and uniqueness}
\label{sec4}
In this section, we present details of the proof of the main theorem following the framework described in Sect.\,\ref{AU}.
\subsection{Locality conditions}
We start by verifying the locality conditions \eqref{BB}. 

Let us first find the differential polynomials $M_\ql^\qz$ and $N_\ql^\qz$ in \eqref{BI} to ensure that $I_0$ is local. By using the fact \eqref{BL} and the equations \eqref{BH} and \eqref{BI}, we arrive at
\begin{equation}
\label{BX}
\fk{\diff{}{\qt_0}}{\diff{}{s_2}}v_\ql = c_0\qs_{\ql,2}^1+a^{\qa\qb}\diff{f_{\qa,1}}{\qt_0}\diff{v_\ql}{t^{\qb,0}}-\qp_x\kk{M_\ql^\qz\qs_{\qz,2}+N_\ql^\qz\qs_{\qz,1}}+loc.
\end{equation}
Here and henceforth we will use $loc$ to denote the local terms, i.e., terms belonging to $\hm A$. We need to find $M_\ql^\qz$ and $N_\ql^\qz$ such that the right hand side of  \eqref{BX} is local.
\begin{Lem}\label{BY}
The following identities hold true:
\[
\diff{}{\qs_{\qb,0}}\vard{P_1}{\qs_{\ql,0}} = \eta^{\qa\qb}\kk{\frac 12+\mu_\qa}\diff{v^\ql}{t^{\qa,0}},\quad \diff{}{\qs_{\qb,0}}\vard{P_1}{v^\ql} = \eta^{\qa\qb}\kk{\frac 12+\mu_\qa}\diff{\qs_{\ql,0}}{t^{\qa,0}}.
\]
\end{Lem}
\begin{proof}
By using the identity \eqref{AN} and the fact that $H_{\qa,-1}$ are Casimirs of $P_0$, we obtain the identity
\[
\fk{P_1}{H_{\qa,-1}} = \kk{\frac 12+\mu_\qa}\fk{P_0}{H_{\qa,0}}.
\]
By taking the variational derivatives of both sides of the identity above and by using \eqref{a2-3} and \eqref{a2-4}, it is easy to prove the desired identities. The lemma is proved.
\end{proof}
\begin{Lem}
We have the following relation:
\begin{equation}
\label{BV}
\qs_{\ql,2}^1 = \kk{\frac 12+\mu_\qe}\eta^{\qz\qe}\diff{v_\ql}{t^{\qe,0}}\qs_{\qz,1}+loc.
\end{equation}
\end{Lem}
\begin{proof}
Let us denote by $\mathcal P_1$ the Hamiltonian operator of $P_1$ and represent it in the form
\[
\mathcal P_1 = \sum_{k\geq 0}\mathcal P_{1,k}\qp_x^k,\quad \mathcal P_{1,k}\in M_n(\hm A).
\]
Then according to the recursion relation
\[
\eta^{\qg\ql}\qs_{\ql,2}^1 = \mathcal P_1^{\qg\ql}\qs_{\ql,1} = \sum_{k\geq 0}\mathcal P_{1,k}^{\qg\ql}\qs_{\ql,1}^k,
\]
it is easy to see that 
\[
\qs_{\ql,2}^1 = \eta_{\ql\qg}\mathcal P_{1,0}^{\qg\qz}\qs_{\qz,1}+loc.
\]
Now it follows from the definition of the Hamiltonian operator that
\[
\mathcal P_{1,0}^{\qg\qz} = \diff{}{\qs_{\qz,0}}\vard{P_1}{\qs_{\qg,0}}.
\]
Therefore by using Lemma \ref{BY} we arrive at the identity \eqref{BV} and the lemma is proved.
\end{proof}
\begin{Prop}
\label{DT}
There exist unique differential polynomials $M^\qz_\ql$ and $N^\qz_\ql$  such that \eqref{BX} is local. More explicitly, we have
\[
M^\qz_\ql = \kk{\frac 52+\mu_\ql+c_0}\qd^\qz_\ql,\quad \qp_xN^\qz_{\ql} = \kk{\frac 12+\mu_\qa}\eta^{\qa\qz}\kk{\mu_\qz-\mu_\ql-1}\diff{v_\ql}{t^{\qa,0}},
\]
here $c_0$ is the arbitrary constant appearing in the operator $\mathcal L_2$.
\end{Prop}
\begin{proof}
To ensure the vanishing of the coefficient of $\qs_{\qz,2}$ in the right hand side of \eqref{BX},  $M^\qz_\ql$ should be a constant, so it is determined by the leading terms of the right hand sides of \eqref{BH} and \eqref{BI}, which are fixed by the Virasoro symmetry $\diff{}{s_2}$ of the super tau-cover of the Principal Hierarchy. Hence we obtain
\[M^\qz_\ql = \kk{\frac 52+\mu_\ql+c_0}\qd^\qz_\ql.
\]

By using Proposition \ref{BM} we have
\begin{equation}
\label{CB}
\kk{\frac 12+\mu_\qa}\diff{f_{\qa,1}}{\qt_0} = \qs_{\qa,1}+loc.
\end{equation}
Then by using this equation, the recursion relations \eqref{2-11} and the identity \eqref{BV}, we see that the vanishing of the coefficient of $\qs_{\qz,1}$ in the right hand side of \eqref{BX} gives the equation
\begin{equation}
\label{CG}
\qp_xN^\qz_{\ql} = \kk{\frac 12+\mu_\qa}\eta^{\qa\qz}\kk{\mu_\qz-\mu_\ql-1}\diff{v_\ql}{t^{\qa,0}}.
\end{equation}
Both sides of \eqref{CG} are total $x$-derivatives, so we can integrate \eqref{CG} to obtain $N^\qz_{\ql}$ upto a constant,  which has differential degree zero. Since the differential degree zero part of $N^\qz_{\ql}$ is fixed by genus zero Virasoro symmetry, $N^\qz_{\ql}$ is uniquely determined. The proposition is proved.
\end{proof}

Now by a direct computation of the condition
\[
\fk{\diff{}{\qt_0}}{\diff{}{s_2}} = 0,
\]
we obtain the following explicit expression for $I_0$ which is defined by \eqref{BD}:
\begin{align}
\label{BZ}
I_0v_\ql =\, &a^{\qa\qb}\kk{f_{\qb,0}'\diff{\qs_{\ql,0}}{t^{\qa,1}}+f_{\qa,1}'\diff{\qs_{\ql,0}}{t^{\qb,0}}}-a^{\qa\qb}\kk{\diff{f_{\qb,0}}{\qt_0}\diff{v_\ql}{t^{\qa,1}}+\diff{f_{\qa,1}}{\qt_0}\diff{v_\ql}{t^{\qb,0}}}\\
\notag
&+b^{\qa\qb}\kk{f_{\qb,0}'\diff{\qs_{\ql,0}}{t^{\qa,0}}+f_{\qa,0}'\diff{\qs_{\ql,0}}{t^{\qb,0}}}-b^{\qa\qb}\kk{\diff{f_{\qb,0}}{\qt_0}\diff{v_\ql}{t^{\qa,0}}+\diff{f_{\qa,0}}{\qt_0}\diff{v_\ql}{t^{\qb,0}}}\\
\notag
&+\kk{\frac 52+\mu_\ql}\qs_{\ql,2}^1+W\qs_{\ql,0}+\qp_x\kk{N^\qz_\ql\qs_{\qz,1}},\\
\label{CA}
I_0\qs_{\ql,0} =\, &-a^{\qa\qb}\kk{\diff{f_{\qb,0}}{\qt_0}\diff{\qs_{\ql,0}}{t^{\qa,1}}+\diff{f_{\qa,1}}{\qt_0}\diff{\qs_{\ql,0}}{t^{\qb,0}}}
-b^{\qa\qb}\kk{\diff{f_{\qb,0}}{\qt_0}\diff{\qs_{\ql,0}}{t^{\qa,0}}+\diff{f_{\qa,0}}{\qt_0}\diff{\qs_{\ql,0}}{t^{\qb,0}}}\\
\notag
&+\kk{\frac 52+\mu_\ql}\diff{\qs_{\ql,0}}{\qt_2}-\diff{}{\qt_0}\kk{N^\qz_\ql\qs_{\qz,1}},
\end{align}
here and henceforth we will use $W$ to denote the coefficient of $t^{1,0}$ of the operator $\mathcal L_2$, more explicitly, we have
\begin{align*}
W =&\, \kk{\frac 12+\mu_1}\kk{\frac 32+\mu_1}\kk{\frac 52+\mu_1}\diff{}{t^{1,2}}\\
&+\sum_{k=1}^2\kk{3\kk{\frac 12+\mu_1}^2+6\kk{\frac 12+\mu_1}+2}\kk{R_k}^\qb_1 \diff{}{t^{\qb,2-k}}\\
&+\kk{\frac 92+3\mu_1}\kk{R_{2,2}}^\qb_1\diff{}{t^{\qb,0}}.
\end{align*}
\begin{Prop}
The derivation $I_0$ given by \eqref{BZ}, \eqref{CA} is local.
\end{Prop}
\begin{proof}
The locality of \eqref{BZ} is equivalent to the definition of $N_\ql^\qz$ and we only need to check the locality of \eqref{CA}. By using the definition of the odd flows we know that
\[
\diff{\qs_{\ql,0}}{\qt_2} =T_{0,2}\diff{\qs_{\ql,0}}{\qt_1} = \sum_{s\geq 0}\qs_{\qb,1}^s\diff{}{\qs_{\qb,0}^s}\diff{\qs_{\ql,0}}{\qt_1}.
\]
Therefore it follows from Lemma \ref{BY} that
\[
\diff{\qs_{\ql,0}}{\qt_2} = \qs_{\qb,1}\eta^{\qa\qb}\kk{\frac 12+\mu_\qa}\diff{\qs_{\ql,0}}{t^{\qa,0}}+loc.
\]
Then the locality of \eqref{CA} is proved by using \eqref{CB} and the following obvious fact:
\[
\diff{N^\qz_{\ql}}{\qt_0} = \kk{\frac 12+\mu_\qa}\eta^{\qa\qz}\kk{\mu_\qz-\mu_\ql-1}\diff{\qs_{\ql,0}}{t^{\qa,0}}.
\]
The proposition is proved.
\end{proof}

Next let us prove the locality of $I_1$. Similar to the expression \eqref{BZ} and \eqref{CA}, we can write down the explicit expression for $I_1$, which can be found in the next subsection. However the locality of $I_1$ can not be derived from this expression directly, so we turn to prove the following equivalent conditions:
\begin{equation}
\label{DP}
\fk{\diff{}{\qt_1}}{\diff{}{s_2}}v_\ql\in\hm A,\quad \fk{\diff{}{\qt_1}}{\diff{}{s_2}}\qs_{\ql,0}\in\hm A.
\end{equation}

By a direct computation, we have
\begin{align}
\label{CC}
\fk{\diff{}{\qt_1}}{\diff{}{s_2}}v_\ql =&\ a^{\qa\qb}\kk{\diff{f_{\qb,0}}{\qt_1}\diff{v_\ql}{t^{\qa,1}}+\diff{f_{\qa,1}}{\qt_1}\diff{v_\ql}{t^{\qb,0}}}\\
\notag
&+b^{\qa\qb}\kk{\diff{f_{\qb,0}}{\qt_1}\diff{v_\ql}{t^{\qa,0}}+\diff{f_{\qa,0}}{\qt_1}\diff{v_\ql}{t^{\qb,0}}}+(1+c_0)\diff{v_\ql}{\qt_3}\\
\notag
&-\sum_{s\geq 0}\kk{\kk{\frac 52+c_0+\mu_\qz}\qs_{\qz,2}^s+\qp_x^s\kk{N^\qe_\qz\qs_{\qz,1}}}\diff{}{\qs_{\qz,0}^s}\diff{v_\ql}{\qt_1}+loc.
\end{align}
\begin{Lem}
\label{CE}
There exists a unique differential polynomial $Z^\qb_{\qa}$ such that
\[
\kk{\frac 12+\mu_\qa}\diff{f_{\qa,1}}{\qt_1} = \qs_{\qa,2}+Z^\qb_{\qa}\qs_{\qb,1}+loc,
\]
where $Z^\qb_{\qa}$ satisfies the following equation
\begin{equation}
\label{CD}
\qp_xZ^\qb_{\qa} = -\kk{\frac 12+\mu_\qe}\eta^{\qb\qe}\diff{v_\qa}{t^{\qe,0}}.
\end{equation}
\end{Lem}
\begin{proof}
The existence and uniqueness of $Z^\qb_\qa$ can be obtained from Proposition \ref{BM}, hence we only need to derive \eqref{CD}, which can be obtained by using the identity \eqref{BV} and the fact that
$
\qp_x\kk{\qs_{\qa,2}+Z^\qb_{\qa}\qs_{\qb,1}}
$
is local. The lemma is proved.
\end{proof}
The following equation can be derived from the definition of the odd flows and the recursion relation \eqref{2-11}:
\[
\diff{v_\ql}{\qt_3} = \sum_{s\geq 0}\qs_{\qz,2}^s\diff{}{\qs_{\qz,0}^s}\diff{v_\ql}{\qt_1}.
\]
This equation together with Lemma \ref{CE} and the identity \eqref{BV} enables us to rewrite \eqref{CC} in the form
\[
\fk{\diff{}{\qt_1}}{\diff{}{s_2}}v_\ql = U_\ql^\qz\qs_{\qz,2}+V_\ql^\qz\qs_{\qz,1}+loc,
\]
where the differential polynomials $U_\ql^\qz$ and $V_\ql^\qz$ are given by
\[
U_\ql^\qz = \eta^{\qz\qb}\kk{\frac 12+\mu_\qb}\kk{\frac 32+\mu_\qz}\diff{v_\ql}{t^{\qb,0}}-\kk{\frac 32+\mu_\qz}\diff{}{\qs_{\qz,0}}\diff{v_\ql}{\qt_1},
\]
\[
V_\ql^\qz = a^{\qa\qz}\diff{v_\ql}{t^{\qa,1}}+(b^{\qa\qz}+b^{\qz\qa})\diff{v_{\ql}}{t^{\qa,0}}-\sum_{s\geq 0}\qp_x^s\kk{N^\qz_{\qg}-\kk{\frac 32+\mu_\qg}Z^\qz_{\qg}}\diff{}{\qs_{\qg,0}^s}\diff{v_\ql}{\qt_1}.
\]
The vanishing of $U_\ql^\qz$ is a consequence of Lemma
\ref{BY} and the vanishing of $V_\ql^\qz$ is due to the following lemma.
\begin{Lem}
\label{CF}
We have $V_\ql^\qz = 0$, that is 
\[
\sum_{s\geq 0}\qp_x^s\kk{N^\qz_{\qg}-\kk{\frac 32+\mu_\qg}Z^\qz_{\qg}}\diff{}{\qs_{\qg,0}^s}\diff{v_\ql}{\qt_1} = a^{\qa\qz}\diff{v_\ql}{t^{\qa,1}}+(b^{\qa\qz}+b^{\qz\qa})\diff{v_{\ql}}{t^{\qa,0}}.
\]
\end{Lem}
\begin{proof}
Consider the functional
\[
F^\qz = \eta^{\qa\qz}\kk{\frac 12+\mu_\qa}\kk{\frac 12+\mu_\qz}\int h_{\qa,1},
\]
then by using \eqref{CG} and \eqref{CD}, we can check that there exist constants $C^\qz_\ql$ such that
\[
N^\qz_{\ql}-\kk{\frac 32+\mu_\qg}Z^\qz_{\ql} = \vard{F^\qz}{v_\ql}+C^\qz_\ql.
\]
Then by using \eqref{a2-4}, we obtain
\[
\sum_{s\geq 0}\qp_x^s\kk{N^\qz_{\qg}-\kk{\frac 32+\mu_\qg}Z^\qz_{\qg}}\diff{}{\qs_{\qg,0}^s}\diff{v_\ql}{\qt_1} = -\vard{}{\qs_{\ql,0}}\fk{F^\qz}{P_1}+C^\qz_\qg\diff{}{\qs_{\qg,0}}\diff{v_\ql}{\qt_1}.
\]
Using the bihamiltonian recursion relations \eqref{AN} and Lemma \ref{BY} we see that the right hand side of the above equation is a linear combination of the flows $\diff{v_\ql}{t^{\qa,p}}$, hence so is $V^\qz_\ql$. The genus zero Virasoro symmetry implies that $V^\qz_\ql\in\hm A^0_{\geq 2}$, so by using the theory of bihamiltonian cohomology \cite{DLZ-1}, we know that $V^\qz_\ql$ must vanish.
The lemma is proved.
\end{proof}

We have verified the first relation in \eqref{DP} and let us proceed to prove the second one. We have
\begin{align}
\label{CH}
\fk{\diff{}{\qt_1}}{\diff{}{s_2}}\qs_{\ql,0} =&\ a^{\qa\qb}\kk{\diff{f_{\qb,0}}{\qt_1}\diff{\qs_{\ql,0}}{t^{\qa,1}}+\diff{f_{\qa,1}}{\qt_1}\diff{\qs_{\ql,0}}{t^{\qb,0}}}\\
\notag
&+b^{\qa\qb}\kk{\diff{f_{\qb,0}}{\qt_1}\diff{\qs_{\ql,0}}{t^{\qa,0}}+\diff{f_{\qa,0}}{\qt_1}\diff{\qs_{\ql,0}}{t^{\qb,0}}}\\
\notag
&+\kk{\frac 52+c_0+\mu_\ql}\diff{\qs_{\ql,2}}{\qt_1}+\diff{N_\ql^\qz}{\qt_1}\qs_{\qz,1}+(1+c_0)\diff{\qs_{\ql,0}}{\qt_3}\\
\notag
&-\sum_{s\geq 0}\kk{\kk{\frac 52+c_0+\mu_\qz}\qs_{\qz,2}^s+\qp_x^s\kk{N^\qe_\qz\qs_{\qz,1}}}\diff{}{\qs_{\qz,0}^s}\diff{\qs_{\ql,0}}{\qt_1}+loc.
\end{align}
By using the equations
\[
\diff{\qs_{\ql,0}}{\qt_3} = \sum_{s\geq 0}\qs_{\qg,2}^s\diff{}{\qs_{\qg,0}^s}\diff{\qs_{\ql,0}}{\qt_1}+\diff{\qs_{\ql,1}}{\qt_2},\quad \diff{\qs_{\ql,1}}{t^{\qa,p}} = \diff{}{\qt_1}\vard{h_{\qa,p+1}}{v^\ql},
\]
\[
\diff{\qs_{\ql,1}}{\qt_2} = \qs_{\qz,1}\kk{\frac 12+\mu_\qe}\eta^{\qe\qz}\diff{}{\qt_1}\vard{h_{\qe,1}}{v^\ql}+loc
\]
we obtain that
\[
\fk{\diff{}{\qt_1}}{\diff{}{s_2}}\qs_{\ql,0} = \qs_{\qz,2}\tilde U_\ql^\qz+\qs_{\qz,1}\tilde V_\ql^\qz+loc,
\]
where the differential polynomials $\tilde U_\ql^\qz$, $\tilde V_\ql^\qz$ are given by
\begin{align*}
\tilde U_\ql^\qz =&\, \eta^{\qz\qb}\kk{\frac 12+\mu_\qb}\kk{\frac 32+\mu_\qz}\diff{\qs_{\ql,0}}{t^{\qb,0}}-\kk{\frac 32+\mu_\qz}\diff{}{\qs_{\qz,0}}\diff{\qs_{\ql,0}}{\qt_1},\\
\tilde V_\ql^\qz =&\, a_{\qa\qz}\diff{\qs_{\ql,0}}{t^{\qa,1}}+(b_{\qa\qz}+b_{\qz\qa})\diff{\qs_{\ql,0}}{t^{\qa,0}}-\diff{}{\qt_1}\kk{N^\qz_{\ql}-\kk{\frac 32+\mu_\qg}Z^\qz_{\ql}}\\&-\sum_{s\geq 0}\qp_x^s\kk{N^\qz_{\qg}-\kk{\frac 32+\mu_\qg}Z^\qz_{\qg}}\diff{}{\qs_{\qg,0}^s}\diff{\qs_{\ql,0}}{\qt_1}.
\end{align*}
These differential polynomials actually vanish, the proof of which is similar to the one for the vanishing of $U_\ql^\qz$ and $V_\ql^\qz$ given above. Hence the locality condition \eqref{BB} is verified.

Finally let us consider the locality condition \eqref{AZ}. The first condition
\[\fk{\diff{}{t^{\qd,j}}}{\diff{}{s_2}}v_\ql\in\hm A\] 
is obvious. On the other hand, we have
\begin{align}
\label{CI}
\fk{\diff{}{t^{\qd,j}}}{\diff{}{s_2}}\qs_{\ql,0} =&\ \kk{\frac 52+c_0+\mu_\ql}\diff{\qs_{\ql,2}}{t^{\qd,j}}+\diff{N_\ql^\qz}{t^{\qd,j}}\qs_{\qz,1}\\
\notag
&-\sum_{s\geq 0}\qp_x^s\kk{N^\qz_{\qg}\qs_{\qz,1}+\kk{\frac 52+c_0+\mu_\qg}\qs_{\qg,2}}\diff{}{\qs_{\qg,0}^s}\diff{\qs_{\ql,0}}{t^{\qd,j}}+loc.
\end{align}

\begin{Prop}
The right hand side of the equation \eqref{CI} is local.
\end{Prop}
\begin{proof}
By using the following identity proved in \cite{liu2011jacobi}
\[
\diff{}{\qs_{\qb,0}}\vard{}{v^\qa} = \vard{}{v^\qa}\vard{}{\qs_{\qb,0}}
\]
we immediately obtain that
\[
\diff{}{\qs_{\qb,0}}\diff{\qs_{\ql,0}}{t^{\qa,p}} = \diff{}{\qs_{\qb,0}}\vard{X_{\qa,p}}{v^\ql} = 0.
\]
This means that the flows $\diff{\qs_{\ql,2}}{t^{\qd,j}}$ can be written as
\[
\diff{\qs_{\ql,2}}{t^{\qd,j}} =T_2\diff{\qs_{\ql,0}}{t^{\qd,j}} = \sum_{s\geq 0}\qs_{\qb,2}^s\diff{}{\qs_{\qb,0}^s}\diff{\qs_{\ql,0}}{t^{\qd,j}}=  B^\qz_\ql\qs_{\qz,1}+loc,
\]
where $B_\ql^\qz$ are certain differential polynomials. By using the identity \eqref{BV}, we know that 
\[
B^\qz_\ql = -\diff{Z_\ql^\qz}{t^{\qd,j}}.
\]
By using \eqref{BV} again, we can represent the right hand side of \eqref{CI} in the form
\[
V^\qz_\ql\qs_{\qz,1}+loc,
\]
where $V^\qz_{\qd,j;\ql}$ is the following differential polynomial:
\begin{align*}
V^\qz_{\qd,j;\ql} =&\, \diff{}{t^{\qd,j}}\kk{N^\qz_{\ql}-\kk{\frac 52+c_0+\mu_\qg}Z^\qz_{\ql}}\\
&-\sum_{s\geq 0}\qp_x^s\kk{N^\qz_{\qg}-\kk{\frac 52+c_0+\mu_\qg}Z^\qz_{\qg}}\diff{}{\qs_{\qg,0}^s}\diff{\qs_{\ql,0}}{t^{\qd,j}}.
\end{align*}

Define the following functional
\[
F^\qz = \eta^{\qa\qz}\kk{\frac 12+\mu_\qa}\kk{\frac 32+c_0+\mu_\qz}\int h_{\qa,1},
\]
then by applying \eqref{a2-3} we have the following identity:
\[
V^\qz_{\qd,j;\ql} = \vard{}{v^\ql}\fk{F^\qz}{X_{\qd,j}},
\]
here $X_{\qd,j}$ is the vector field defined in \eqref{AD}. Since the functional $F^\qz$ is a conserved quantity of the deformed Principal hierarchy, hence $\fk{F^\qz}{X_{\qd,j}} = 0$. The proposition is proved.
\end{proof}

\subsection{Closedness conditions}
In this subsection we prove the closedness condition \eqref{BC} and \eqref{BF}. The verification of \eqref{BC} is straightforward by using the explicit expression \eqref{BZ} and \eqref{CA}, so we omit the details here. 

Let us proceed to prove the closedness condition \eqref{BF}. We fix a choice of $X^\circ\in\derx$ such that $X^\circ$ satisfies 
\[
\fk{\diff{}{\qt_0}}{X^\circ} = I_0
\]
and that the differential degree zero part of $X^\circ$ is given by genus zero Virasoro symmetry. 

We first write down the explicit expression for $I_1$ defined in \eqref{BD}. Define a derivation $\hat I_1\in\der^0$ by the  formulae
$
\hat I_1v_\qg = \hat I_1\qs_{\qg,0} = 0
$
and 
\begin{align*}
\hat I_1v_\qg^{(n)} = & n\qp_x^{n-1}W(v_\qg)+a^{\qa\qb}\qp_x^n\kk{f_{\qb,0}\diff{v_\qg}{t^{\qa,1}}+f_{\qa,1}\diff{v_\qg}{t^{\qb,0}}}+b^{\qa\qb}\qp_x^n\kk{f_{\qb,0}\diff{v_\qg}{t^{\qa,0}}+f_{\qa,0}\diff{v_\qg}{t^{\qb,0}}}\\
&-a^{\qa\qb}\kk{f_{\qb,0}\qp_x^n\diff{v_\qg}{t^{\qa,1}}+f_{\qa,1}\qp_x^n\diff{v_\qg}{t^{\qb,0}}}-b^{\qa\qb}\kk{f_{\qb,0}\qp_x^n\diff{v_\qg}{t^{\qa,0}}+f_{\qa,0}\qp_x^n\diff{v_\qg}{t^{\qb,0}}},\\
\hat I_1\qs_{\qg,0}^n = & n\qp_x^{n-1}W(\qs_{\qg,0})+a^{\qa\qb}\qp_x^n\kk{f_{\qb,0}\diff{\qs_{\qg,0}}{t^{\qa,1}}+f_{\qa,1}\diff{\qs_{\qg,0}}{t^{\qb,0}}}+b^{\qa\qb}\qp_x^n\kk{f_{\qb,0}\diff{\qs_{\qg,0}}{t^{\qa,0}}+f_{\qa,0}\diff{\qs_{\qg,0}}{t^{\qb,0}}}\\
&-a^{\qa\qb}\kk{f_{\qb,0}\qp_x^n\diff{\qs_{\qg,0}}{t^{\qa,1}}+f_{\qa,1}\qp_x^n\diff{\qs_{\qg,0}}{t^{\qb,0}}}-b^{\qa\qb}\kk{f_{\qb,0}\qp_x^n\diff{\qs_{\qg,0}}{t^{\qa,0}}+f_{\qa,0}\qp_x^n\diff{\qs_{\qg,0}}{t^{\qb,0}}}\\
&+\kk{\frac 32+\mu_\qg}\kk{\qs_{\qg,2}^n+\qp_x^nZ^\qe_\qg\qs_{\qe,1}}+\qp_x^n\kk{N^\qe_\qg\qs_{\qe,1}}-\qp_x^nN^\qe_\qg\qs_{\qe,1},
\end{align*}
here $n\geq 1$. Note that this derivation does NOT commute with $\qp_x$. It is easy to see that $\hat I_1$ is indeed local. By a careful computation, we obtain the following expression for $I_1$:
\begin{align*}
I_1v_\ql =&\, \hat I_1\kk{\diff{v_\ql}{\qt_1}}+A^\qa\diff{v_\ql}{t^{\qa,1}}+B^\qa\diff{v_\ql}{t^{\qa,0}},\\
I_1\qs_{\ql,0} =&\, \hat I_1\kk{\diff{\qs_{\ql,0}}{\qt_1}}+A^\qa\diff{\qs_{\ql,0}}{t^{\qa,1}}+B^\qa\diff{\qs_{\ql,0}}{t^{\qa,0}}\\
&+\kk{\frac 32+\mu_\ql}\kk{{\diff{\qs_{\ql,1}}{\qt_2}-\qs_{\qb,1}\eta^{\qa\qb}\kk{\frac 12+\mu_\qa}\diff{\qs_{\ql,1}}{t^{\qa,0}}}},
\end{align*}
where $A^\qa$ and $B^\qa$ are local differential polynomials given by
\begin{align*}
A^\qa = & -a^{\qa\qb}\kk{\diff{f_{\qb,0}}{\qt_1}-\qs_{\qb,1}},\\
B^\qa = & -\eta^{\qa\qb}\kk{\frac 12 +\mu_\qa}\kk{\frac 32+\mu_\qb}\kk{\kk{\frac 12+\mu_\qb}\diff{f_{\qb,1}}{\qt_1}-\qs_{\qb,2}-Z^\qe_\qb\qs_{\qe,1}}\\
&-(b^{\qa\qb}+b^{\qb\qa})\kk{\diff{f_{\qb,0}}{\qt_1}-\qs_{\qb,1}}.
\end{align*}

We start by proving the first identity in \eqref{BF}, which can also be written as
\[\fk{\diff{}{\qt_0}}{I_1}+\fk{\diff{}{\qt_1}}{\fk{\diff{}{\qt_0}}{X^\circ}} = 0.\]
The following lemmas are obtained by lengthy but straightforward computations.
\begin{Lem}
\label{CM}
Let $Q\in\hm A$ be any local differential polynomial, then we have
\begin{align*}
\fk{\hat I_1}{\qp_x}Q =&\ W(Q)+a_{\qa\qb}\kk{f_{\qb,0}'\diff{Q}{t^{\qa,1}}+f_{\qa,1}'\diff{Q}{t^{\qb,0}}}+ b_{\qa\qb}\kk{f_{\qb,0}'\diff{Q}{t^{\qa,0}}+f_{\qa,0}'\diff{Q}{t^{\qb,0}}}\\
&+\sum_{n\geq 0}\qs_{\qz,1}^1\qp_x^n\kk{Y^\qz_{\qg}-\kk{\frac 32+\mu_\qg}Z^\qz_{\qg}}\diff{}{\qs_{\qg,0}^n}Q\\&+\kk{\frac 32+\mu_\qg}\qp_x\kk{\qs_{\qg,2}+Z^\qe_\qg\qs_{\qe,1}}\diff{}{\qs_{\qg,0}}Q.
\end{align*}
\end{Lem}
\begin{Lem}
\label{CJ}
The following identities hold true:
\begin{align*}
\fk{\diff{}{\qt_0}}{I_1}v_\ql =&\ \fk{\diff{}{\qt_1}}{\fk{\diff{}{\qt_0}}{\hat I_1}}v_\ql+\diff{A^\qa}{\qt_0}\diff{v_\ql}{t^{\qa,1}}+\diff{B^\qa}{\qt_0}\diff{v_\ql}{t^{\qa,0}},\\
\fk{\diff{}{\qt_0}}{I_1}\qs_{\ql,0} =&\ \fk{\diff{}{\qt_1}}{\fk{\diff{}{\qt_0}}{\hat I_1}}\qs_{\ql,0}+\diff{A^\qa}{\qt_0}\diff{\qs_{\ql,0}}{t^{\qa,1}}+\diff{B^\qa}{\qt_0}\diff{\qs_{\ql,0}}{t^{\qa,0}}\\
&+\kk{\frac 32+\mu_\ql}\kk{{\diff{\qs_{\ql,1}}{\qt_2}-\qs_{\qb,1}\eta^{\qa\qb}\kk{\frac 12+\mu_\qa}\diff{\qs_{\ql,1}}{t^{\qa,0}}}}.
\end{align*}
\end{Lem}
\begin{Lem}
\label{CK}
We have the following decomposition:
\[\fk{\diff{}{\qt_0}}{\hat I_1+X^\circ} = D+\diff{}{\qt_2},\]
where $D$ is a derivation $\hm A\to\hm A^+$ given by the  formulae
\begin{align*}
&\,Dv_\qg^{(n)} \\=&\, -a^{\qa\qb}\kk{\diff{f_{\qb,0}}{\qt_0}\qp_x^n\diff{v_\qg}{t^{\qa,1}}+\diff{f_{\qa,1}}{\qt_0}\qp_x^n\diff{v_\qg}{t^{\qb,0}}}-b^{\qa\qb}\kk{\diff{f_{\qb,0}}{\qt_0}\qp_x^n\diff{v_\ql}{t^{\qa,0}}+\diff{f_{\qa,0}}{\qt_0}\qp_x^n\diff{v_\ql}{t^{\qb,0}}}\\
&+\qs_{\qb,1}\eta^{\qa\qb}\kk{\frac 12+\mu_\qa}\kk{\frac 12+\mu_\qb}\qp_x^n\diff{v_\qg}{t^{\qa,0}},
\end{align*}
\begin{align*}
&\,D\qs_{\qg,0}^n \\=&\, -a^{\qa\qb}\kk{\diff{f_{\qb,0}}{\qt_0}\qp_x^n\diff{\qs_{\qg,0}}{t^{\qa,1}}+\diff{f_{\qa,1}}{\qt_0}\qp_x^n\diff{\qs_{\qg,0}}{t^{\qb,0}}}-b^{\qa\qb}\kk{\diff{f_{\qb,0}}{\qt_0}\qp_x^n\diff{\qs_{\qg,0}}{t^{\qa,0}}+\diff{f_{\qa,0}}{\qt_0}\qp_x^n\diff{\qs_{\qg,0}}{t^{\qb,0}}}\\
&+\qs_{\qb,1}\eta^{\qa\qb}\kk{\frac 12+\mu_\qa}\kk{\frac 12+\mu_\qb}\qp_x^n\diff{\qs_{\qg,0}}{t^{\qa,0}}+\diff{\qs_{\qb,0}}{\qt_1}\qp_x^n\kk{N^\qb_\qg-\kk{\frac 32+\mu_\qg}Z^\qb_\qg}\\
&-\qd_{n,0}\kk{\frac 32+\mu_\qg}\diff{}{\qt_0}\kk{\qs_{\qg,2}+Z^\qe_\qg\qs_{\qe,1}}.
\end{align*}
\end{Lem}
\begin{Prop}
The first identity of the closedness condition \eqref{BF} holds true, i.e., 
\[\fk{\diff{}{\qt_0}}{I_1}+\fk{\diff{}{\qt_1}}{\fk{\diff{}{\qt_0}}{X^\circ}} = 0.\]
\end{Prop}
\begin{proof}
It follows from Lemma \ref{CJ} that
\begin{align*}
&\,\fk{\diff{}{\qt_0}}{I_1}v_\ql+\fk{\diff{}{\qt_1}}{\fk{\diff{}{\qt_0}}{X^\circ}}v_\ql\\=&\,
\fk{\diff{}{\qt_1}}{\fk{\diff{}{\qt_0}}{\hat I_1+X^\circ}}v_\ql+\diff{A^\qa}{\qt_0}\diff{v_\ql}{t^{\qa,1}}+\diff{B^\qa}{\qt_0}\diff{v_\ql}{t^{\qa,0}}.
\end{align*}
Due to Lemma \ref{CK} we only need to verify the following identity: 
\[
\fk{\diff{}{\qt_1}}{D}v_\ql+\diff{A^\qa}{\qt_0}\diff{v_\ql}{t^{\qa,1}}+\diff{B^\qa}{\qt_0}\diff{v_\ql}{t^{\qa,0}} = 0,
\]
which can be checked directly from the definition of $D$. Similarly we can prove that
\[
\fk{\diff{}{\qt_0}}{I_1}\qs_{\ql,0}+\fk{\diff{}{\qt_1}}{\fk{\diff{}{\qt_0}}{X^\circ}}\qs_{\ql,0} = 0.
\]
The proposition is proved.
\end{proof}

It remains to prove the second identity of \eqref{BF}.
\begin{Lem}
\label{CL}
The following identities hold true:
\begin{align*}
\fk{I_1}{\diff{}{\qt_1}}v_\ql =&\, \sum_{s\geq 0}G^\qg_s\diff{}{v_\qg^{(s)}}\diff{v_\ql}{\qt_1}+F^\qg_s\diff{}{\qs_{\qg,0}^{s}}\diff{v_\ql}{\qt_1}\\
&-A^\qa\diff{}{\qt_1}\diff{v_\ql}{t^{\qa,1}}-B^\qa\diff{}{\qt_1}\diff{v_\ql}{t^{\qa,0}}+\diff{B^\qa}{\qt_1}\diff{v_\ql}{t^{\qa,0}},\\
\fk{I_1}{\diff{}{\qt_1}}\qs_{\ql,0} =&\, \sum_{s\geq 0}G^\qg_s\diff{}{v_\qg^{(s)}}\diff{\qs_{\ql,0}}{\qt_1}+F^\qg_s\diff{}{\qs_{\qg,0}^{s}}\diff{\qs_{\ql,0}}{\qt_1}\\
&-A^\qa\diff{}{\qt_1}\diff{\qs_{\ql,0}}{t^{\qa,1}}-B^\qa\diff{}{\qt_1}\diff{\qs_{\ql,0}}{t^{\qa,0}}+\diff{B^\qa}{\qt_1}\diff{\qs_{\ql,0}}{t^{\qa,0}},
\end{align*}
where $G^\qg_s$ and $F^\qg_s$ are differential polynomials defined by
\[
G_s^\qg = \fk{\diff{}{\qt_1}}{\hat I_1}v_\qg^{(s)}+\qp_x^s\kk{I_1v_\qg},\quad F_s^\qg = \fk{\diff{}{\qt_1}}{\hat I_1}\qs_{\qg,0}^s+\qp_x^s\kk{I_1\qs_{\qg,0}}.
\]
\end{Lem}
\begin{Lem}
\label{CN}
The differential polynomials $G^\qg_s$ and $F^\qg_s$ defined in Lemma \ref{CL} have the following expressions:
\begin{align*}
G_s^\qg =&\, A^\qa\qp_x^s\diff{v_\qg}{t^{\qa,1}}+B^\qa\qp_x^s\diff{v_{\qg}}{t^{\qa,0}},\\
F_s^\qg =&\, A^\qa\qp_x^s\diff{\qs_{\qg,0}}{t^{\qa,1}}+B^\qa\qp_x^s\diff{\qs_{\qg,0}}{t^{\qa,0}}\\
&+\qd_{s,0}\kk{\frac 32+\mu_\qg}\kk{\diff{\qs_{\qg,1}}{\qt_2}-\qs_{\qb,1}\eta^{\qa\qb}\kk{\frac 12+\mu_\qa}\diff{\qs_{\qg,1}}{t^{\qa,0}}}.
\end{align*}
\end{Lem}
\begin{proof}
By using the definition of $G^\qg_s$ we can obtain the identity
\[
G^\qg_{s+1}-\qp_xG^\qg_s = \diff{}{\qt_1}\kk{\fk{\hat I_1}{\qp_x}v_\qg^{(s)}}+\fk{\qp_x}{\hat I_1}\diff{v_\qg^{(s)}}{\qt_1}.
\]
Therefore it follows from Lemma \ref{CM} that
\[
G^\qg_{s+1}-\qp_xG^\qg_s = -\qp_xA^\qa\qp_x^s\diff{v_\qg}{t^{\qa,1}}-\qp_xB^\qa\qp_x^s\diff{v_\qg}{t^{\qa,0}}.
\]
Hence $G^\qg_s$ can be solved recursively starting from the initial condition
\[
G^\qg_0 = I_1v_\qg-\hat I_1\diff{v_\qg}{\qt_1} = A^\qa\diff{v_\qg}{t^{\qa,1}}+B^\qa\diff{v_\qg}{t^{\qa,0}}.
\]
The differential polynomials $F_s^\qg$ can be computed similarly. The lemma is proved.
\end{proof}
\begin{Prop}
The second identity of the closedness condition \eqref{BF} holds true, i.e., 
\[\fk{\diff{}{\qt_1}}{I_1} = 0.\]
\end{Prop}
\begin{proof}
The proof of the proposition is straightforward by combining the results of Lemma \ref{CL} and Lemma \ref{CN}.
\end{proof}

\subsection{Vanishing of the genus one obstruction}In this subsection, we will work with the canonical coordinates of the Frobenius manifold. Let us start by recalling some useful formulae related to the canonical coordinates. For the details one may refer to the work \cite{dubrovin1996geometry,dubrovin2001normal}.

We denote by $u^1,\cdots,u^n$ the local canonical coordinates of a semisimple Frobenius manifold $M$ and denote by $(u^i;\qth_i)$ the corresponding coordinates of $\hat M$. Under this system of local coordinates, the bihamiltonian structure \eqref{biham-frob} takes the form \eqref{CW}, i.e.,
\[
P_0^{[0]} = \frac 12\int \sum_{i,j=1}^n\kk{\qd_{i,j} f^i\qth_i\qth_i^1+ A^{ij}\qth_i\qth_j},\quad P_1^{[0]} = \frac 12\int \sum_{i,j=1}^n\kk{\qd_{i,j} u^if^i\qth_i\qth_i^1+ B^{ij}\qth_i\qth_j}.
\]
Introduce functions 
\begin{equation}
\label{CP}
\psi_{i1} = \frac{1}{\sqrt{f^i}},
\end{equation}
where the sign of the square root can be arbitrarily chosen, and define
\[
\psi_{i\qa} = \frac{1}{\psi_{i1}}\diff{v_\qa}{u^i},\quad \qg_{ij} = \frac{1}{\psi_{j1}}\diff{\psi_{i1}}{u^j},\quad i\neq j.
\]
Then it is proved in \cite{dubrovin1996geometry} that
\begin{equation}
\label{CQ}
\diff{v_\qa}{u^i} = \psi_{i1}\psi_{i\qa},\quad \diff{u^i}{v^\qa} = \frac{\psi_{i\qa}}{\psi_{i1}},
\end{equation}
\begin{equation}
\label{CO}
\diff{\psi_{i\qa}}{u^k} = \qg_{ik}\psi_{k\qa},\quad i\neq k;\quad \diff{\psi_{i\qa}}{u^i} = -\sum_{k\neq i}\qg_{ik}\psi_{k\qa}.
\end{equation}
Let 
$
V_{ij} = \mu_\qa\eta^{\qa\qb}\psi_{i\qa}\psi_{j\qb},
$
then we have the identity 
\begin{equation}
\label{CR}
\qg_{ij}(u^j-u^i) = V_{ij}.
\end{equation}

Now we are going to describe the deformation problem given in Sect. \ref{AU} in terms of the canonical coordinates. Introduce the odd variables $\qth_{i,m}$ by the recursion relations \eqref{2-9}, then it follows from the the relation \eqref{2-10} that the equations \eqref{BH} and \eqref{BI} can be represented in the form
\begin{align}
\label{DR}
\diff{u^i}{s_2} =\ & a^{\qa\qb}\kk{f_{\qb,0}\diff{u^i}{t^{\qa,1}}+f_{\qa,1}\diff{u^i}{t^{\qb,0}}}+ b^{\qa\qb}\kk{f_{\qb,0}\diff{u^i}{t^{\qa,0}}+f_{\qa,0}\diff{u^i}{t^{\qb,0}}}\\\notag&+Xu^i+\mathcal L_2u^i,\\
\label{DS}
\diff{\qth_{i,0}}{s_2} =\ & a^{\qa\qb}\kk{f_{\qb,0}\diff{\qth_{i,0}}{t^{\qa,1}}+f_{\qa,1}\diff{\qth_{i,0}}{t^{\qb,0}}}+ b^{\qa\qb}\kk{f_{\qb,0}\diff{\qth_{i,0}}{t^{\qa,0}}+f_{\qa,0}\diff{\qth_{i,0}}{t^{\qb,0}}}\\\notag&+X\qth_{i,0}
+\sum_jA^j_i\qth_{j,2}+B^j_i\qth_{j,1}+\mathcal L_2\qth_{i,0}.
\end{align}

In the canonical coordinates, the Hamiltonian operator $\mathcal P_1$ of $P_1$ has the form
\[
\mathcal P_1^{ij} =u^if^i\qd_{ij}\qp_x+\frac 12 \qp_x\kk{u^if^i}\qd_{ij}+B^{ij}+Q^{ij}\qp_x^3+\cdots
\]
where $B^{ij}$ is defined in \eqref{DQ} and $Q^{ij}$ is given by the following formula (c.f., e.g., \cite{falqui2012exact})
\begin{equation}
\label{CV}
Q^{ij} = 3c_i\kk{f^i}^2\qd_{ij}+\frac 12\kk{u^i-u^j}\kk{f^j\qp_jf^ic_i-f^i\qp_if^jc_j},
\end{equation}
here $c_i$ is the $i$-th central invariant. Let us also denote the evolutions of $\qth_i$ along the flows $\diff{}{t^{\qa,p}}$ by the following equations:
\[
\diff{\qth_i}{t^{\qa,p}} = T_{\qa,p}^i\qth_i^1+\cdots+K_{\qa,p;j}^i\qth_j^3+\cdots,\quad  T_{\qa,p}\in\hm A^0_0,\  K_{\qa,p;j}^i\in\hm A^0_3,
\]
here and henceforth we omit the terms that do not contribute to the relevant computations given later. Note that the coefficients of $\qth_j^1$ of the leading term of $\diff{\qth_i}{t^{\qa,p}}$ are zero for $j\neq i$, this is due to the fact that the leading term of the flow $\diff{}{t^{\qa,p}}$ is diagonal, one may refer to \cite{dubrovin2018bihamiltonian} for details.

Let us turn to the proof of the vanishing of the genus one obstruction. Due to the closedness condition \eqref{BC}, we can  choose a derivation $X^\circ$ such that when we take $X = X^\circ$ in \eqref{DR} and \eqref{DS} we have
\[
\fk{\diff{}{\qt_0}}{\diff{}{s_2}} = 0.
\]
We require that the leading term of $X^\circ$ is given by the genus zero Virasoro symmetry, hence it follows from the result of \cite{dubrovin1999frobenius} that the actions of $X^\circ$ on $u^i$ and $\qth_i$ take the form
\begin{align*}
X^\circ u^i &= (u^i)^3+\sum_jL^i_ju^{j,2}+\cdots,\quad L^i_j\in\hm A^0_0, \\
 X^\circ \qth_i &= \sum_jM^j_i\qth_j+\sum_jJ^j_i\qth_j^2+\cdots,\quad M^j_i,J^j_i\in\hm A^0_0.
\end{align*}
\begin{Lem}
\label{CS}
We have the following identity:
\begin{align*}
J^i_i-L^i_i =&\, \frac{1}{f^i}2c_0(u^i)^2Q^{ii}-\sum_j\frac{A^j_i}{f^j}(u^i+u^j)Q^{ij}-\sum_j\frac{B^j_i}{f^j}Q^{ij}\\
&-b^{\qb,q}_{2;1,0}K^i_{\qb,q;i}-a^{\qa,p;\qb,q}_2\kk{\kk{f_{\qb,q}'}_0K^i_{\qa,p;i}+\kk{f_{\qa,p}'}_0K^i_{\qb,q;i}},
\end{align*}
where $b^{\qb,q}_{2;\qa,p}$ and $a^{\qa,p;\qb,q}_2$ are the constants that appear in the operator $\mathcal L_2$, and $\kk{f_{\qa,p}'}_0$ denotes the differential degree zero component of $f_{\qa,p}'$.
\end{Lem}
\begin{proof}
We can prove this lemma by computing the coefficient of $\qth_i^3$ of differential degree 3 component of the left hand side of the equation
\[
\fk{\diff{}{\qt_0}}{\diff{}{s_2}}u^i = 0.
\]
The lemma is proved.
\end{proof}
\begin{Rem}
Note that $c_0$ is the arbitrary constant that appears in the operator $\mathcal L_2$, and $c_1,\cdots,c_n$ are the central invariants.
\end{Rem}
According to the discussion given in Sect.\,\ref{AU}, we need to compute the cohomology class of the differential degree 3 component of the derivation $\fk{\diff{}{\qt_1}}{\mathcal C} = I_1-\fk{\diff{}{\qt_1}}{X^\circ}$. Due to Theorem \ref{AR}, it suffices to prove that the coefficient of  $\qth_i^3$ of the differential degree 3 component of
\[
I_1u^i-\fk{\diff{}{\qt_1}}{X^\circ}u^i
\]
vanishes.
By using Lemma \ref{CS} and a straightforward computation, we obtain the following lemma.
\begin{Lem}
The coefficient of  $\qth_i^3$ of the differential degree 3 component of
$
I_1u^i-\fk{\diff{}{\qt_1}}{X^\circ}u^i
$ reads
\begin{align}
\label{CT}
&\sum_{j}Q^{ij}\kk{M^i_j+A^i_j(u^i)^2+B^i_ju^i}+E^3Q^{ii}-(6+c_0)(u^i)^2Q^{ii}\\
\notag
&+3Q^{ii}b^{\qb,q}_{2;1,0}T_{\qb,q}^i+3Q^{ii}a^{\qa,p;\qb,q}_2\kk{\kk{f_{\qb,q}'}_0T_{\qa,p}^i+\kk{f_{\qa,p}'}_0T_{\qb,q}^i},
\end{align}
here $E^3$ is the cubic of the Euler vector field, which is given by
\[
E^3 = \sum_i(u^i)^3\diff{}{u^i}.
\]
\end{Lem}
\begin{Lem}
\label{CU}
We have the identities
\begin{align*}
&M^i_i+A^i_i(u^i)^2+B^i_iu^i\\ =&\, (3+c_0)(u^i)^2-\frac{1}{f^i}E^3f^i-b^{\qb,q}_{2;1,0}T_{\qb,q}^i-a^{\qa,p;\qb,q}_2\kk{\kk{f_{\qb,q}'}_0T_{\qa,p}^i+\kk{f_{\qa,p}'}_0T_{\qb,q}^i},
\end{align*}
and $M^j_i+A^j_i(u^j)^2+B^j_iu^j = 0$ for $i\neq j$.
\end{Lem}
\begin{proof}
We can prove this lemma by computing the coefficient of $\qth_j^1$ of the differential degree 1 component of the left hand side of the equation
\[
\fk{\diff{}{\qt_0}}{\diff{}{s_2}}u^i = 0.
\]
The lemma is proved.
\end{proof}
In order to prove the vanishing of the genus one obstruction, we need to verify that the expression \eqref{CT} vanishes. To this end, we only need to check, due to Lemma \ref{CU} and the expression \eqref{CV} for $Q^{ii}$, the following identity:
\[
E^3f^i+2f^i\kk{b^{\qb,q}_{2;1,0}T_{\qb,q}^i+a^{\qa,p;\qb,q}_2\kk{\kk{f_{\qb,q}'}_0T_{\qa,p}^i+\kk{f_{\qa,p}'}_0T_{\qb,q}^i}} = 3(u^i)^2f^i.
\]
\begin{Prop}
For $m\geq -1$, we have
\[
E^{m+1}f^i+2f^i\kk{b^{\qb,q}_{m;1,0}T_{\qb,q}^i+a^{\qa,p;\qb,q}_m\kk{\kk{f_{\qb,q}'}_0T_{\qa,p}^i+\kk{f_{\qa,p}'}_0T_{\qb,q}^i}} = (1+m)(u^i)^mf^i,
\]
here $E^{m+1}$ is the $(m+1)$-th power of the Euler vector field $E$ which is given by
\[
E^{m+1} = \sum_i(u^i)^{m+1}\diff{}{u^i}.
\]
\end{Prop}
\begin{proof}
Let us consider the following generating functions:
\begin{align*}
&\sum_{m\geq -1}\frac{1}{\ql^{m+2}}E^{m+1}f^i = \sum_j\frac{1}{\ql-u^j}\diff{f^i}{u^j};\\
&\sum_{m\geq -1}\frac{1}{\ql^{m+2}}(1+m)(u^i)^mf^i = \frac{f^i}{(\ql-u^i)^2};\\
&\sum_{m\geq -1}\frac{1}{\ql^{m+2}}\kk{b^{\qb,q}_{m;1,0}T_{\qb,q}^i+a^{\qa,p;\qb,q}_m\kk{\kk{f_{\qb,q}'}_0T_{\qa,p}^i+\kk{f_{\qa,p}'}_0T_{\qb,q}^i}}\\ 
=&\, \frac 12\frac{1}{(\ql-u^i)^2}+\sum_{j\neq i}\frac{\psi_{j1}}{\psi_{i1}}\frac{V_{ij}}{(\ql-u^i)(\ql-u^j)},
\end{align*}
where the last generating function was computed in \cite{dubrovin2001normal}, one may refer to Lemma 3.10.18 and the proof of Theorem 3.10.29 of \cite{dubrovin2001normal} for details. Then the proposition is proved by using \eqref{CP}, \eqref{CO} and \eqref{CR}.
\end{proof}
Finally we have the following theorem.
\begin{Th}
\label{DD}
There exists a unique $X\in\derx^0$ such that  the following flow defines a symmetry of the deformed Principal Hierarchy which is a deformation of the genus zero Virasoro symmetry $\diff{}{s_2}$ of the Principal Hierarchy:
\begin{align}
\label{CX}
\diff{v_\ql}{s_2} =&\,  a^{\qa\qb}\kk{f_{\qb,0}\diff{v_\ql}{t^{\qa,1}}+f_{\qa,1}\diff{v_\ql}{t^{\qb,0}}}+ b^{\qa\qb}\kk{f_{\qb,0}\diff{v_\ql}{t^{\qa,0}}+f_{\qa,0}\diff{v_\ql}{t^{\qb,0}}}\\\notag&+Xv_{\ql}+\mathcal L_2v_\ql,\\
\label{CY}
\diff{\qs_{\ql,0}}{s_2} =&\,  a^{\qa\qb}\kk{f_{\qb,0}\diff{\qs_{\ql,0}}{t^{\qa,1}}+f_{\qa,1}\diff{\qs_{\ql,0}}{t^{\qb,0}}}+ b^{\qa\qb}\kk{f_{\qb,0}\diff{\qs_{\ql,0}}{t^{\qa,0}}+f_{\qa,0}\diff{\qs_{\ql,0}}{t^{\qb,0}}}\\\notag&+X\qs_{\ql,0}
+\kk{\frac 52+c_0+\mu_\ql}\qs_{\ql,2}+N^\qz_{\ql}\qs_{\qz,1}+\mathcal L_2\qs_{\ql,0},
\end{align}
where $N^\qz_{\ql}\in\hm A^0_{\geq 0}$ is the differential polynomial described in Lemma \ref{DT}.
\end{Th}

\subsection{Lifting to the tau-covers}
In order to lift the symmetry \eqref{CX} to the tau-cover of the deformed Principal Hierarchy, we first need to rewrite \eqref{CX} in terms of the normal coordinates $w^1,\cdots, w^n$ of $M$. We start by proving the following lemmas.
\begin{Lem}
\label{CZ}
Let $g_\ql\in\hm A^0_{\geq 1}$ be the differential polynomials given in \eqref{AM} which satisfy the identities
\[
h_{\ql,0} = v_\ql+\qp_xg_\ql,\quad \ql = 1,\cdots,n.
\]
Then we have:
\begin{enumerate}
\item $g_1 = 0$.
\item For any $\qa = 1,\cdots,n$ and $p\geq 0$,
\[
\Qo_{\qa,p;\ql,0} = \vard{h_{\qa,p+1}}{v^\ql}+\diff{g_\ql}{t^{\qa,p}}.
\]
\end{enumerate}
\end{Lem}
\begin{proof}
Due to Proposition \ref{AF}, we have $D_{X_{1,0}} = \qp_x$ and $X_{1,0} = -[H_{1,0},P_0]$, from which it follows that
\[
v_\qa  = \vard{H_{1,0}}{v^\qa}.
\]
By taking $\qa = 1$, we obtain the first property by using the definition \eqref{AE}. The second one is obvious due to Theorem \ref{AI}. The lemma is proved.
\end{proof}
\begin{Lem}
There exists a derivation $X^\circ\in\derx$ such that its leading term is given by the Virasoro symmetry $\diff{}{s_2}$ of the Principal Hierarchy, and it satisfies the equations $\fk{\diff{}{\qt_0}}{X^\circ} = I_0$ and 
\begin{align}
\label{DA}
X^\circ h_{\ql,0}+\hat I_1h_{\ql,0} =&\, a^{\qa\qb}\kk{f_{\qa,1}'\Qo_{\qb,0;\ql,0}+f_{\qb,0}'\Qo_{\qa,1;\ql,0}+\frac{\qp^2\Qo_{\qa,1;\qb,0}}{\qp t^{\ql,0}\qp t^{1,0}}}\\\notag
&+b^{\qa\qb}\kk{f_{\qa,0}'\Qo_{\qb,0;\ql,0}+f_{\qb,0}'\Qo_{\qa,0;\ql,0}+\frac{\qp^2\Qo_{\qa,0;\qb,0}}{\qp t^{\ql,0}\qp t^{1,0}}}\\\notag
&+b^{\qb,q}_{2;1,0}\Qo_{\qb,q;\ql,0}+b^{\qb,q}_{2;\ql,0}\Qo_{\qb,q;1,0}+2c_{2;\ql,0;1,0}.
\end{align}
\end{Lem}
\begin{proof}
We first find a particular solution $\tilde X^\circ$ of the equation $\fk{\diff{}{\qt_0}}{X} = I_0$, then we modify it by a solution $\tilde{\mathcal C}$ of the homogeneous equation $\fk{\diff{}{\qt_0}}{{\mathcal C}} = 0$ such that $X^\circ:=\tilde X^\circ+\tilde{\mathcal C}$ satisfies \eqref{DA}.

Let us define $\tilde X^\circ\in\derx$ as follows:
\begin{align*}
\tilde X^\circ v_\ql =&\ a^{\qa\qb}\kk{f_{\qa,1}'\vard{h_{\qb,1}}{v^\ql}+f_{\qb,0}'\vard{h_{\qa,2}}{v^\ql}}+b^{\qa\qb}\kk{f_{\qa,0}'\vard{h_{\qb,1}}{v^\ql}+f_{\qb,0}'\vard{h_{\qa,1}}{v^\ql}}\\
&+b^{\qb,q}_{2;1,0}\vard{h_{\qb,q+1}}{v^\ql}+b^{\qb,q}_{2;\ql,0}\vard{h_{\qb,q+1}}{v^1}+2c_{2;\ql,0;1,0};\\
\tilde X^\circ \qs_{\ql,0} =&\ a^{\qa\qb}\kk{\diff{f_{\qa,1}}{\qt_0}\vard{h_{\qb,1}}{v^\ql}+\diff{f_{\qb,0}}{\qt_0}\vard{h_{\qa,2}}{v^\ql}}+b^{\qa\qb}\kk{\diff{f_{\qa,0}}{\qt_0}\vard{h_{\qb,1}}{v^\ql}+\diff{f_{\qb,0}}{\qt_0}\vard{h_{\qa,1}}{v^\ql}}\\
&+b^{\qb,q}_{2;\ql,0}\diff{f_{\qb,q}}{\qt_0}-\kk{\frac 52+\mu_\ql}\qs_{\ql,2}-N^\qz_\ql\qs_{\qz,1}.
\end{align*}

Firstly, from the definition of $N^\qz_\ql$ it follows that $\tilde X^\circ$ is indeed local. We also note that the leading terms of $\tilde X^\circ v_\ql$ and $\tilde X^\circ \qs_{\ql,0}$ coincide with the local terms of the Virasoro symmetry $\diff{v_\ql}{s_2}$ and $\diff{\qs_{\ql,0}}{s_2}$ of the Principal Hierarchy. By a direct computation, it is easy to check that $\fk{\diff{}{\qt_0}}{\tilde X^\circ} = I_0$.

Next we want to determine $\tilde{\mathcal C}\in\derx$ such that $X^\circ = \tilde X^\circ+\tilde{\mathcal C}$ satisfies \eqref{DA}. By using the definition of $\tilde X^\circ v_\ql$ and Lemma \ref{CM}, \ref{CZ}, it is straightforward to show that
\begin{equation}
\label{DB}
\tilde{\mathcal C}(h_{\ql,0}) = a^{\qa\qb}\frac{\qp^2\Qo_{\qa,1;\qb,0}}{\qp t^{\ql,0}\qp t^{1,0}} + b^{\qa\qb}\frac{\qp^2\Qo_{\qa,0;\qb,0}}{\qp t^{\ql,0}\qp t^{1,0}}-\qp_x\hat I_1(g_\ql)-\qp_x\tilde X^\circ(g_\ql),
\end{equation}
which uniquely determines the actions $\tilde{\mathcal C}v_\ql$.

Finally we need to check such $\tilde{\mathcal C}$ satisfies $\fk{\diff{}{\qt_0}}{\tilde{\mathcal C}} = 0$. By using the next lemma, we know that it suffices to show that $\int \tilde{\mathcal C} v_\ql$ = 0, which is obvious from \eqref{DB} and \eqref{AM}. The lemma is proved.
\end{proof}
\begin{Lem}
\label{DC}
Let $U_1,\cdots, U_n$ be differential polynomials with $U_\ql\in\hm A^0_{\geq 2}$. Then there exists $\mathcal C\in\derx^0_{\geq 1}$ such that $\fk{\diff{}{\qt_0}}{\mathcal C} = 0$ and $\mathcal C v_\ql = U_\ql$ if and only if $\int U_\ql = 0$.
\end{Lem}
\begin{proof}
Let $\mathcal C\in\derx^0$ be a derivation such that $\fk{\diff{}{\qt_0}}{\mathcal C} = 0$. Then by using the triviality of the variational Hamiltonian cohomology $H^0_{\geq 1}\bigl(\derx\bigr)$, we know the existence of a certain $\mathcal K\in\derx^{-1}$ such that $\fk{\diff{}{\qt_0}}{\mathcal K} = \mathcal C$. Let us denote $\mathcal K\qs_{\ql,0} = V_\ql\in\hm A^0$. Then it is easy to see that
\[
\mathcal Cv_\ql = \fk{\diff{}{\qt_0}}{\mathcal K}v_\ql = \mathcal K\qs_{\ql,0}^1 = \qp_xV_\ql.
\]
Therefore we have $\int U_\ql = \int \qp_xV_\ql = 0$.

Conversely, if $U_\ql = \qp_xV_\ql$ for some $V_\ql\in\hm A^0$, we can define a derivation $\mathcal C\in\derx^0$ by
\[
\mathcal C v_\ql = \qp_xV_\ql,\quad \mathcal C\qs_{\ql,0} = \diff{V_\ql}{\qt_0},
\]
then it is easy to check that $\fk{\diff{}{\qt_0}}{\mathcal C} = 0$ and the lemma is proved.
\end{proof}

Now let us denote $\mathcal C = X - X^\circ$, where the derivation $X$ is described in Theorem \ref{DD} and $X^\circ$ satisfies \eqref{DA}. Then we can rewrite the Virasoro symmetry in terms of the normal coordinates as follows:
\begin{align*}
\diff{w_\ql}{s_2} =&\  a^{\qa\qb}\kk{f_{\qb,0}\diff{w_\ql}{t^{\qa,1}}+f_{\qa,1}\diff{w_\ql}{t^{\qb,0}}}+ b^{\qa\qb}\kk{f_{\qb,0}\diff{w_\ql}{t^{\qa,0}}+f_{\qa,0}\diff{w_\ql}{t^{\qb,0}}}\\
&+a^{\qa\qb}\kk{f_{\qa,1}'\Qo_{\qb,0;\ql,0}+f_{\qb,0}'\Qo_{\qa,1;\ql,0}+\frac{\qp^2\Qo_{\qa,1;\qb,0}}{\qp t^{\ql,0}\qp t^{1,0}}}\\
&+b^{\qa\qb}\kk{f_{\qa,0}'\Qo_{\qb,0;\ql,0}+f_{\qb,0}'\Qo_{\qa,0;\ql,0}+\frac{\qp^2\Qo_{\qa,0;\qb,0}}{\qp t^{\ql,0}\qp t^{1,0}}}\\
&+b^{\qb,q}_{2;1,0}\Qo_{\qb,q;\ql,0}+b^{\qb,q}_{2;\ql,0}\Qo_{\qb,q;1,0}+2c_{2;\ql,0;1,0}+\mathcal C(h_{\ql,0})+\mathcal L_2 w_\ql.
\end{align*}
According to Lemma \ref{DC}, there exists $Q_\ql\in\hm A^0_{\geq 1}$ such that $\mathcal C(h_{\ql,0}) = \qp_x Q_\ql$. In particular, the fact that
\[
\diff{}{t^{\qg,0}}\diff{w_\ql}{s_2} = \diff{}{t^{\ql,0}}\diff{w_\qg}{s_2}
\]
gives the equation
\[
\diff{Q_\ql}{t^{\qg,0}} = \diff{Q_\qg}{t^{\ql,0}},
\]
which means that $\int Q_1$ is a conserved quantity of the flows $\diff{}{t^{\ql,0}}$. If the dimension of the Frobenius manifold $\dim M\geq 2$, then it follows from Lemma 4.12 and Theorem A.2 of \cite{dubrovin2018bihamiltonian} that there exists $Q\in\hm A^0_{\geq 0}$ such that $Q_\ql = \diff{Q}{t^{\ql,0}}$. 
\begin{Th}
\label{DF}
Let $\mathcal Z$ be a tau-function of the tau-cover \eqref{AW} of the deformed Principal Hierarchy. Then there exists a differential polynomial $O_2\in\hm A^0_{\geq 0}$ which yields a symmetry of the tau-cover given by
\[
\diff{\mathcal Z}{s_2} = L_2^{even}\mathcal Z+O_2\mathcal Z,
\]
where the evolutions of $f_{\qa,p}$ and $w_\ql$ along the flow $\diff{}{s_2}$ are defined by
\[
\diff{f_{\qa,p}}{s_2} = \diff{}{t^{\qa,p}}\diff{\log Z}{s_2},\quad \diff{w_\ql}{s_2} = \diff{}{t^{1,0}}\diff{f_{\ql,0}}{s_2}.
\]
\end{Th}
\begin{proof}
When $\dim M\geq 2$ we take $O_2 = Q$ which is just given above. When $\dim M = 1$, $O_2$ is given in \eqref{DE}. The theorem is proved.
\end{proof}

Now let us proceed to determine all other Virasoro symmetries $\diff{}{s_m}$ of the tau-cover of the deformed Principal Hierarchy. In what follows we will denote by $\mathcal F = \log\mathcal Z$ and denote
\begin{align*}
&\mathcal G_m = a_m^{\qa,p;\qb,q}\kk{\diff{\mathcal F}{t^{\qa,p}}\diff{\mathcal F}{t^{\qb,q}}+\frac{\qp^2\mathcal F}{\qp t^{\qa,p}\qp t^{\qb,q}}}+c_{m;\qa,p;\qb,q}t^{\qa,p}t^{\qb,q};\\&\mathcal L_m^{even} = b^{\qb,q}_{m;\qa,p}t^{\qa,p}\diff{}{t^{\qb,q}},\quad m\geq -1.
\end{align*}
It is proved in \cite{dubrovin2018bihamiltonian} that
\[
\diff{\mathcal F}{s_{-1}} = \mathcal G_{-1}+\mathcal L_{-1}^{even}\mathcal F
\]
induces a symmetry of the tau-cover of the deformed Principal Hierarchy. On the other hand, by using Theorem \ref{DF} we know that
\[
\diff{\mathcal F}{s_{2}} = \mathcal G_{2}+O_2+\mathcal L_{2}^{even}\mathcal F
\]
also induces a symmetry of the tau-cover of the deformed Principal Hierarchy. From the commutation relation
\[
[L_{-1}^{even},L_2^{even}] = -3 L_1^{even},
\]
it follows that
\[
\fk{\diff{}{s_{-1}}}{\diff{}{s_2}}\mathcal F = 3\kk{\mathcal G_1+\mathcal L_1^{even}\mathcal F}+\diff{O_2}{s_{-1}}-\mathcal L_{-1}^{even}O_2.
\]
Since 
$
\diff{O_2}{s_{-1}}-\mathcal L_{-1}^{even}O_2
$
is a differential polynomial, we can define
\[
\diff{\mathcal F}{s_1} =\frac 13\fk{\diff{}{s_{-1}}}{\diff{}{s_2}}\mathcal F =  \mathcal G_1+O_1+\mathcal L_1^{even}\mathcal F,\quad O_1 =\frac 13\kk{\diff{O_2}{s_{-1}}-\mathcal L_{-1}^{even}O_2}.
\]
Then it is obvious that $\diff{\mathcal F}{s_1}$ induces a symmetry of the tau-cover of the deformed Principal Hierarchy. Similarly we define the symmetry
\[
\diff{\mathcal F}{s_0} = \frac12 \fk{\diff{}{s_{-1}}}{\diff{}{s_1}}\mathcal F = \mathcal G_0+O_0+\mathcal L_0^{even}\mathcal F,\quad O_0 =\frac 12\kk{\diff{O_1}{s_{-1}}-\mathcal L_{-1}^{even}O_1}.
\]
Now let us define $\diff{\mathcal F}{s_m}$ for $m\geq 3$ recursively in the following way. Assume we have defined 
\[
\diff{\mathcal F}{s_m} = \mathcal G_m+O_m+\mathcal L_m^{even}\mathcal F
\]
such that it induces a symmetry of the tau-cover of the deformed Principal Hierarchy for $m\geq 2$. Then we have
\[
\diff{}{s_1}\diff{\mathcal F}{s_m} = \diff{\mathcal G_m}{s_1}+\diff{O_m}{s_1}+\mathcal L_m^{even}\kk{\mathcal G_1+O_1+\mathcal L_1^{even}\mathcal F}.
\]
It follows from the definition of $\mathcal G_m$ that
\begin{align*}
\diff{\mathcal G_m}{s_1} =&\, \diff{}{s_1}\kk{a_m^{\qa,p;\qb,q}\kk{\diff{\mathcal F}{t^{\qa,p}}\diff{\mathcal F}{t^{\qb,q}}+\frac{\qp^2\mathcal F}{\qp t^{\qa,p}\qp t^{\qb,q}}}+c_{m;\qa,p;\qb,q}t^{\qa,p}t^{\qb,q}}\\
=&\, a_m^{\qa,p;\qb,q}\kk{f_{\qa,p}\diff{}{t^{\qb,q}}(\mathcal G_1+O_1+\mathcal L_1^{even}\mathcal F)+f_{\qb,q}\diff{}{t^{\qa,p}}(\mathcal G_1+O_1+\mathcal L_1^{even}\mathcal F)}\\
&+a_m^{\qa,p;\qb,q}\frac{\qp^2}{\qp t^{\qa,p}\qp t^{\qb,q}}(\mathcal G_1+O_1+\mathcal L_1^{even}\mathcal F)\\
=&\, a_m^{\qa,p;\qb,q}\kk{f_{\qa,p}\diff{O_1}{t^{\qb,q}}+f_{\qb,q}\diff{O_1}{t^{\qa,p}}+\frac{\qp^2O_1}{\qp t^{\qa,p}\qp t^{\qb,q}}}+\cdots,
\end{align*}
here $\cdots$ stands for remaining terms that do not contain $O_1$. In a similar way we can compute $\diff{}{s_m}\diff{\mathcal F}{s_1}$. By using the commutation relation
\[
[L_1^{even},L_m^{even}] = (1-m)L_{m+1}^{even}
\]
we obtain the following result:
\begin{align*}
\fk{\diff{}{s_1}}{\diff{}{s_m}}\mathcal F =&\ (m-1)\kk{\mathcal G_{m+1}+\mathcal L_{m+1}^{even}\mathcal F}\\
&+\diff{O_m}{s_1}-a_1^{\qa,p;\qb,q}\kk{f_{\qa,p}\diff{O_m}{t^{\qb,q}}+f_{\qb,q}\diff{O_m}{t^{\qa,p}}+\frac{\qp^2O_m}{\qp t^{\qa,p}\qp t^{\qb,q}}}-\mathcal L_1^{even} O_m\\
&-\diff{O_1}{s_m}+a_m^{\qa,p;\qb,q}\kk{f_{\qa,p}\diff{O_1}{t^{\qb,q}}+f_{\qb,q}\diff{O_1}{t^{\qa,p}}+\frac{\qp^2O_1}{\qp t^{\qa,p}\qp t^{\qb,q}}}+\mathcal L_m^{even} O_1.
\end{align*}
Therefore we obtain the symmetry $\diff{\mathcal F}{s_{m+1}}$ by defining
\[
\diff{\mathcal F}{s_{m+1}}=\frac{1}{m-1}\fk{\diff{}{s_1}}{\diff{}{s_m}}\mathcal F = \mathcal G_{m+1}+O_{m+1}+\mathcal L_{m+1}^{even}\mathcal F,
\]
where $O_{m+1}$ is the differential polynomial given by
\begin{align*}
O_{m+1} =&\ \frac{1}{m-1}\kk{\diff{O_m}{s_1}-a_1^{\qa,p;\qb,q}\kk{f_{\qa,p}\diff{O_m}{t^{\qb,q}}+f_{\qb,q}\diff{O_m}{t^{\qa,p}}+\frac{\qp^2O_m}{\qp t^{\qa,p}\qp t^{\qb,q}}}-\mathcal L_1^{even} O_m}\\
&-\frac{1}{m-1}\kk{\diff{O_1}{s_m}-a_m^{\qa,p;\qb,q}\kk{f_{\qa,p}\diff{O_1}{t^{\qb,q}}+f_{\qb,q}\diff{O_1}{t^{\qa,p}}+\frac{\qp^2O_1}{\qp t^{\qa,p}\qp t^{\qb,q}}}-\mathcal L_m^{even} O_1}.
\end{align*}
Thus we have defined recursively an infinite set of symmetries of the tau-cover of the deformed Principal Hierarchy. Their actions on $\mathcal F$ can be represented by
\[
\diff{\mathcal F}{s_m} = \mathcal G_m+O_m+\mathcal L_m^{even}\mathcal F.
\]

Next we show that we can further adjust $O_m$ by adding certain  constants such that these symmetries satisfy the Virasoro commutation relation
\[
\fk{\diff{}{s_k}}{\diff{}{s_l}} = (l-k)\diff{}{s_{k+l}},\quad k,l\geq -1.
\]

\begin{Lem}
There is a unique choice of constants $\kappa_m$ for $m\geq -1$ such that the flows 
\[
\diff{\mathcal F}{s_m} = \mathcal G_m+O_m+\kappa_m+\mathcal L_m^{even}\mathcal F
\]
satisfy the Virasoro commutation relation
\begin{equation}
\label{DH}
\fk{\diff{}{s_k}}{\diff{}{s_l}} = (l-k)\diff{}{s_{k+l}},\quad k,l\geq -1.
\end{equation}
\end{Lem}
\begin{proof}
Let us first fix an arbitrary choice of $O_m$ and denote
\[
\diff{\mathcal F}{\hat s_m} = \mathcal G_m+O_m+\mathcal L_m^{even}\mathcal F.
\]
Then we obtain the differential polynomials $\tilde O_{k+l}$ such that
\[
\fk{\diff{}{\hat s_k}}{\diff{}{\hat s_l}} \mathcal F = (l-k)\kk{\mathcal G_{l+k}+\tilde O_{l+k}+\mathcal L_{l+k}^{even}\mathcal F}.
\]
But both $\diff{}{\hat s_{l+k}}$ and $\fk{\diff{}{\hat s_k}}{\diff{}{\hat s_l}}$ are symmetries of the tau-cover of the deformed Principal Hierarchy, so we conclude that
\[
\diff{\mathcal F}{s}:= \diff{\mathcal F}{\hat s_{k+l}}-\frac{1}{l-k}\fk{\diff{}{\hat s_k}}{\diff{}{\hat s_l}}\mathcal F = O_{k+l}-\tilde O_{k+l}
\]
is also a symmetry of the tau-cover of the deformed Principal Hierarchy. The action of this symmetry on the normal coordinates has the expression
\[
\diff{w_\ql}{s} = \frac{\qp^2}{\qp t^{\ql,0}\qp t^{1,0}}\kk{O_{k+l}-\tilde O_{k+l}}.
\]
Thus $\diff{w_\ql}{s}\in\hm A_{\geq 2}$, and therefore such a symmetry must vanish due to the result of the bihamiltonian cohomology \cite{DLZ-1}. Hence we conclude that
\[
 O_{k+l}-\tilde O_{k+l} = c_{k,l}
\]
for some constant $c_{k,l}$, and this means that
\begin{equation}
\label{DG}
\fk{\diff{}{\hat s_k}}{\diff{}{\hat s_l}} = (l-k)\diff{}{\hat s_{k+l}} + c_{k,l},\quad k,l\geq -1.
\end{equation}

Let us denote by $\mathfrak W_1$ the Lie algebra of formal vector fields on a line, which is an infinite dimensional Lie algebra with a basis
\[
e_m = z^{m+1}\frac{d}{dz},\quad m\geq -1.
\]
Then the relation \eqref{DG} implies that the Lie algebra $\{\diff{}{\hat s_m}\}$ defines a central extension of $\mathfrak W_1$. It is computed in \cite{Gelfand1970} that $H^2(\mathfrak W_1,\mathbb R) = 0$ and hence every central extension is trivial. Therefore we can modify each $O_m$ by adding an appropriate constant $\kappa_m$ such that the modified flows
\[
\diff{\mathcal F}{s_m} = \mathcal G_m+O_m+\kappa_m+\mathcal L_m^{even}\mathcal F
\] 
satisfy the commutation relations \eqref{DH}. Moreover, the choice of $\kappa_m$ is unique since $H^1(\mathfrak W_1,\mathbb R) = 0$ (see \cite{Gelfand1970}). The lemma is proved.
\end{proof}
Thus we have proved the following theorem.
\begin{Th}
For every tau-symmetric bihamiltonian deformation of the Principal Hierarchy associated with a semisimple Frobenius manifold, the deformed integrable hierarchy possesses an infinite set of Virasoro symmetries. The actions of these symmetries on the tau function $\mathcal Z$ are represented by
\[
\diff{\mathcal Z}{s_m} = L_m^{even}\mathcal Z+O_m\mathcal Z,\quad m\geq -1,
\]
where $O_m\in \hm A$ are certain differential polynomials, and the flows $\diff{}{s_m}$ satisfy the commutation relations
\[
\fk{\diff{}{s_k}}{\diff{}{s_l}} = (l-k)\diff{}{s_{k+l}},\quad k,l\geq -1.
\]
\end{Th}
\begin{Ex}
\label{DX}
Let $M$ be the 2-dimensional Frobenius manifold defined on the orbit space of the Weyl group of type $B_2$. Its potential and  Euler vector field are given by
\[
F = \frac 12 v^2u+\frac{4}{15}u^5,\quad E = v\qp_v+\frac 12u\qp_u.
\]
Here $v = v^1$ and $u = v^2$ are the flat coordinates of $M$. We denote by $\qs_1$ and $\qs_2$ the dual coordinates of the fiber of $\hat M$, then the bihamiltonian structure $(P_0^{[0]},P_1^{[0]})$ associated with $M$ is given by
\[
P_0^{[0]} = \frac 12\int \qs_1\qs_2^1+\qs_2\qs_1^1,\quad P_1^{[0]} = \frac 12\int 8u^3\qs_1\qs_1^1+\frac 12 u\qs_2\qs_2^1+2v\qs_1\qs_2^1-\frac 12v_x\qs_1\qs_2.
\]

Let us first write down the Virasoro symmetry $\diff{}{s_1}$ of the tau-cover of the Principal Hierarchy associated with $M$.
The Virasoro operator $L_1^{even}$ has the expression
\[
L_1^{even} = \frac{3}{16}\frac{\qp^2}{\qp t^{1,0}\qp t^{2,0}}+ \mathcal L_1^{even},\]
where\[
\mathcal L_1^{even}=\sum_{p\geq 0}\kk{p+\frac 14}\kk{p+\frac 54}t^{1,p}\diff{}{t^{1,p+1}}+\kk{p+\frac 34}\kk{p+\frac 74}t^{2,p}\diff{}{t^{2,p+1}}.
\]
Then the action of $\diff{}{s_1}$ on the genus zero free energy $\mathcal F^{[0]}$ of the tau-cover of the Principal Hierarchy is given by
\[
\diff{\mathcal F^{[0]}}{s_1} = \frac{3}{16}f_{1,0}f_{2,0}+\mathcal L_1^{even}\mathcal F^{[0]}.
\]

Consider the bihamiltonian structure $(P_0,P_1)$ of the Drinfel'd-Sokolov hierarchy \cite{drinfel1985lie} associated with the untwisted affine Kac-Moody algebra $B_2^{(1)}$. After performing a suitable Miura type transformation we have $P_0 = P_0^{[0]}$, and the Hamiltonian operator $\mathcal P_1$ of $P_1$ has the expression
\[
\mathcal P_1 = \begin{pmatrix}8u^3\qp_x+12u^2u_x & v\qp_x+\frac 14v_x\\ v\qp_x+\frac 34 v_x& \frac 12u\qp_x+\frac 14u_x\end{pmatrix}+\qe^2\begin{pmatrix}D_1 & D_2\\D_3& D_4\end{pmatrix}+O(\qe^4),
\]
where the differential operators $D_i$ are given by $D_2 = u\qp_x^3+\frac 34 u_x\qp_x^2$, $D_4 = \frac 58 \qp_x^3$ and
\begin{align*}
D_1 =&\, 14u^2\qp_x^3+42uu_x\qp_x^2+\kk{20u_x^2+16uu_{xx}+\frac 12v_{xx}}\qp_x+12u_xu_{xx}+6uu^{(3)}+\frac 14 v^{(3)},\\
D_3 =&\,u\qp_x^3+\frac 94u_x\qp_x^2+\frac 32u_{xx}\qp_x+\frac 14u^{(3)}.
\end{align*}
The bihamiltonian structure $(P_0,P_1)$ is a deformation of $(P_0^{[0]},P_1^{[0]})$ with central invariants $c_1 = \frac 16$, $c_2 = \frac{1}{12}$ (see \cite{dubrovin2008frobenius,liu2005deformations}), and it determines a unique deformation of the Principal Hierarchy associated with $M$. 
We can find the Virasoro symmetry $\diff{}{s_1}$ of the tau-cover of the deformed Principal Hierarchy by using the results developed in the present paper. It turns out that the action of $\diff{}{s_1}$ on the tau-function $\mathcal Z$ can be represented by
\begin{equation}
\label{DU}
\diff{\mathcal Z}{s_1} = L_1^{even}\mathcal Z+\kk{\frac 12 u^2+\frac 14v+\frac 14\qe^2 u_{xx}}\mathcal Z.
\end{equation}
A similar result is also given in the Example 5.5 of  \cite{wu2017tau} by using the Kac-Moody-Virasoro algebra.
\end{Ex}
\section{Conclusion}\label{sec5}
In the present paper, we prove the existence of an infinite set of Virasoro symmetries for a given tau-symmetric bihamiltonian deformation of the Principal Hierarchy associated with a semisimple Frobenius manifold. These symmetries can be represented in terms of the tau-function $Z$ of the integrable hierarchy in the form
\begin{equation}
\label{DW}
\diff{Z}{s_m} = L_mZ+O_mZ,\quad m\geq -1.
\end{equation}

Note that the differential polynomials $O_m$ depend on the choice of the representative, in the equivalence class of  Miura type transformations, of the deformations of the bihamiltonian structure of hydrodynamic type $(P_0^{[0]},P_1^{[0]})$. It is proved in \cite{dubrovin2018bihamiltonian} that, for two different choices of representatives $(P_0,P_1)$ and $(\tilde P_0,\tilde P_1)$, the corresponding normal coordinates 
\[
w^\qa =\eta^{\qa\qb}\frac{\qp^2 \log Z}{\qp t^{1,0}\qp t^{\qb,0}},\quad \tilde w^\qa =\eta^{\qa\qb}\frac{\qp^2 \log\tilde Z}{\qp t^{1,0}\qp t^{\qb,0}},\quad \qa=1,\dots,n
\]
of the deformed Principal Hierarchy are related by a Miura type transformation
\[
\tilde w^\qa =w^\qa+\eta^{\qa\qb}\frac{\qp^2 G}{\qp t^{1,0}\qp t^{\qb,0}},
\]
where $G\in\hm A$ is a differential polynomial, and the tau-functions are related by the equation
\begin{equation}\label{DV}\tilde Z = \exp(G)Z.\end{equation}
Conversely, any differential polynomial $G\in\hm A$ defines a Miura type transformation for the deformed bihamiltonian structure and the integrable hierarchy in the manner described above. 

After a Miura type transformation induced from \eqref{DV}, the Virasoro symmetries \eqref{DW} are transformed to the form
\[
\diff{\tilde Z}{s_m} = L_m\tilde Z+\tilde O_m\tilde Z,\quad m\geq -1,
\]
where the differential polynomials $\tilde O_m$ can be computed from $ O_m$ and $G$.

We are going to study the problem of linearization of Virasoro symmetries in subsequent work, i.e., to study whether it is possible to find a suitable differential polynomial $G$ such that all the functuons $\tilde O_m$ vanish. 

Let us exam the possibility of linearizing the Virasoro symmetries given in the example of Sect.\,\ref{B0} for the one-dimensional Frobenius manifold. We want to find a certain Miura type transformation given by \eqref{DV} which linearizes the Virasoro symmetry \eqref{DE} and leaves the expression of the Virasoro symmetry
\[
\diff{Z}{s_{-1}} =  L_{-1}^{even}Z
\]
unchanged. It follows from these requirements that the differential degree zero component $G_0$ of $G$ must satisfy the equations
\[
\diff{G_0}{v} = 0,\quad v^3\diff{G_0}{v} = -\kk{3c-\frac 38}\frac{v^2}{2},
\]
which do not possess any solution unless $c = \frac 18$. When $c = \frac 18$, the linearized Virasoro symmetries for this example are well known \cite{witten1990two}, and the central invariant of the corresponding deformed bihamiltonian structure is $\frac13 c=\frac1{24}$.

We can do a similar computation for Example \ref{DX}. We want to find a Miura type transformation given by \eqref{DV} to linearize the Virasoro symmetry \eqref{DU} and to preserve the expression of the Virasoro symmetry
\[
\diff{Z}{s_{-1}} =  L_{-1}^{even}Z.
\]
Then the differential degree zero component $G_0$ of $G$ must satisfy the equations
\[
\diff{G_0}{v} = 0,\quad uv\diff{G_0}{u} = -\kk{\frac 12 u^2+\frac 14v},
\]
which have no solution. Therefore the Virasoro symmetries given by the bihamiltonian structure $(P_0,P_1)$ in this example cannot be linearized.

In general we have the following conjecture.
\begin{Conj}
The Virasoro symmetries for a given tau-symmetric bihamiltonian deformation of the Principal Hierarchy associated with a semisimple Frobenius manifold is linearizable if and only if the central invariants of the corresponding deformed bihamiltonian structure are all equal to $\frac{1}{24}$.
\end{Conj}
We will study this conjecture and give a proof of it in the paper\cite{linearVir}.


\medskip

\noindent Si-Qi Liu,

\noindent Department of Mathematical Sciences, Tsinghua University \\ 
Beijing 100084, P.R.~China\\
liusq@tsinghua.edu.cn
\medskip

\noindent Zhe Wang,

\noindent Department of Mathematical Sciences, Tsinghua University \\ 
Beijing 100084, P.R.~China\\
zhe-wang17@mails.tsinghua.edu.cn
\medskip

\noindent Youjin Zhang,

\noindent Department of Mathematical Sciences, Tsinghua University \\ 
Beijing 100084, P.R.~China\\
youjin@tsinghua.edu.cn

\end{document}